\documentclass{article}

    \usepackage[preprint]{neurips_2023}

\usepackage[utf8]{inputenc} 
\usepackage[T1]{fontenc}    
\usepackage{hyperref}       
\usepackage{url}            
\usepackage{booktabs}       
\usepackage{amsfonts}       
\usepackage{nicefrac}       
\usepackage{microtype}      
\usepackage{xcolor}         

\usepackage{algorithm}
\usepackage{algorithmic}
\usepackage{booktabs}
\usepackage{subcaption}
\usepackage{apxproof}
\usepackage{adjustbox}
\usepackage[flushleft]{threeparttable}
\usepackage{diagbox}
\newcolumntype{L}[1]{>{\raggedright\let\newline\\\arraybackslash\hspace{0pt}}m{#1}}
\newcolumntype{C}[1]{>{\centering\let\newline\\\arraybackslash\hspace{0pt}}m{#1}}
\newcolumntype{R}[1]{>{\raggedleft\let\newline\\\arraybackslash\hspace{0pt}}m{#1}}

\usepackage[toc,page,header]{appendix}
\usepackage{minitoc}
\doparttoc 
\faketableofcontents 

\usepackage{amsmath}
\usepackage{amssymb}
\usepackage{amsthm}
\usepackage{bm}
\usepackage{mathrsfs}
\usepackage{mathtools}
\usepackage{thm-restate}
\usepackage{thmtools} 
\usepackage{tikz}
\usetikzlibrary{matrix,calc}

\usepackage{caption}

\usepackage{color}

\newcommand{\boldpi}{\boldsymbol{\pi}}

\newcommand{\vtheta}{\bm{\theta}}
\newcommand{\vg}{\bm{g}}
\newcommand{\va}{\bm{a}}
\newcommand{\vpi}{\bm{\pi}}

\newcommand{\cA}{{\mathcal A}}

\newcommand{\cS}{{\mathcal S}}

\newcommand{\cZ}{{\mathcal Z}}

\newcommand{\cN}{\mathcal N}

\newtheorem{theorem}{Theorem}

\newtheorem{lemma}{Lemma}

\newcommand\blfootnote[1]{%
\begingroup
\renewcommand\thefootnote{}\footnote{#1}%
\addtocounter{footnote}{-1}%
\endgroup
}

\title{Taming Multi-Agent Reinforcement Learning with Estimator Variance Reduction}

\author{
Taher Jafferjee\textsuperscript{1}, Juliusz Ziomek\textsuperscript{1}, Tianpei Yang\textsuperscript{2}, Zipeng Dai\textsuperscript{1}, Jianhong Wang\textsuperscript{3}, \\ \textbf{Matthew E. Taylor\textsuperscript{2}, Kun Shao\textsuperscript{1}, Jun Wang\textsuperscript{4}, David Mguni\textsuperscript{1}$^\dag$}\AND
\textsuperscript{1} {\normalfont Huawei Noah's Ark Lab,}
\textsuperscript{2} {\normalfont University of Alberta,} \textsuperscript{3} {\normalfont University of Manchester}, \\ \textsuperscript{\textbf{4}} {\normalfont University College London}}

\begin{document}
\maketitle
\begin{abstract}
Centralised\blfootnote{$^\dag$Corresponding author  <davidmguuni@hotmail.com>. } training with decentralised execution (CT-DE) serves as the foundation of many leading multi-agent reinforcement learning (MARL) algorithms. Despite its popularity, it suffers from a critical drawback due to its reliance on learning from a single sample of the joint-action at a given state. As agents explore and update their policies during training, these single samples may poorly represent the actual joint-policy of the system of agents leading to high variance gradient estimates that hinder learning. To address this problem, we propose an enhancement tool that accommodates any actor-critic MARL method. Our framework, Performance Enhancing Reinforcement Learning Apparatus (PERLA), introduces a sampling technique of the agents' joint-policy into the critics while the agents train. 
This leads to TD updates that closely approximate the true expected value under the current joint-policy rather than estimates from a single sample of the joint-action at a given state. This produces low variance and precise estimates of expected returns, minimising the variance in the critic estimators which typically hinders learning. Moreover, as we demonstrate, by eliminating much of the critic variance from the single sampling of the joint policy, PERLA enables CT-DE methods to scale more efficiently with the number of agents. Theoretically, we prove that PERLA reduces variance in value estimates similar to that of decentralised training while maintaining the benefits of centralised training. Empirically, we demonstrate PERLA's superior performance and ability to reduce estimator variance in a range of benchmarks including \textit{Multi-agent Mujoco}, and \textit{StarCraft II Multi-agent Challenge}.  
\end{abstract}

\section{Introduction} \label{sec:introduction}
Multi-agent reinforcement learning (MARL) has emerged to be a powerful tool to enable autonomous agents to jointly tackle difficult tasks such as ride-sharing \citep{zhou2020smarts} and swarm robotics \citep{mguni2018decentralised}. Nevertheless, a key impediment to these algorithms is the high variance of the critic and policy gradient estimators. Reducing the variance of these estimators is critical since high variance estimators can lead to low sample efficiency and poor overall performance \citep{gu2016q}. 

In multi-agent systems, the environment reward function and state transition dynamics function depend on the \emph{joint-action} of all the agents in the system. As a result, for an agent in a given state, different executions of a particular action may return varying outcomes depending on the joint-actions of other agents. As agents' estimates are based on previous observations of joint actions, updates to the policies of other agents during training may result in returns that significantly deviate from current estimates. Consequently, key estimators for efficient learning can have high variance  severely impairing the learning process and hence, the agents' abilities to jointly maximise performance. 

Centralised Training-Decentralised Execution (CT-DE) paradigm is a popular MARL training framework in which agent's observe the joint behaviour of other agents during training. Therefore, the CT-DE paradigm has at its core a (possibly shared) critic for each agent that makes use of all available information generated by the system, including the global state and the joint action \citep{peng2017multiagent}. This added information can be exploited by the critic during training to promote greater levels of coordination between the agents, which is often required to efficiently learn the optimal joint policies. Moreover, this added information can serve to reduce systemic variance. CT-DE has been shown to be highly effective in promoting high performance outcomes and thus serves as the foundation of many popular MARL methods such as MAPPO, Q-DPP \citep{yang2020multi}, QMIX \citep{rashid2018qmix}, SPOT-AC \citep{mguni2021learning}, and COMA \citep{foerster2018counterfactual}. 

In spite of these benefits, the CT-DE framework can be plagued by high variance updates during training. Central to the learning protocol of CT-DE algorithms are agent policy updates that are based on a single sample of the joint-action executed from other agents' policies at a given state. This can produce value function (VF) updates based on improbable events and result in inaccurate estimates of expected returns. This, in turn often leads to high variance VF estimates and poor sample efficiency. This is exemplified in a simple Coordination Game with the reward structure shown in Figure \ref{fig:concept_experiment}. In this game, miscoordinated actions (i.e., $(l, r)$ or $(r, l)$) are penalised, there is a sub-optimal stable (Nash equilibrium (NE)) joint strategy $(r, r)$, and the optimal joint strategy is $(l, l)$. In this setting, random occurrences of $(l,l)$ are relatively improbable which can induce convergence to the joint strategy $(r, r)$. To illustrate this, suppose the action $(r, r)$ is sampled. A TD update towards this sample (with reward $0.5$)  may cause each agents to increase the policy probability of sampling $r$ (and divert the agents to converge to the sub-optimal NE). On the other hand if the joint-action sampled is $(r, l)$, due to the reward of $-1$, the agent may reduce the policy probability of sampling $r$. Thus, an  update following a single sample of this joint-action leads to an increase in probability weight on $r$, while the other decreases it producing the possibility of highly variant updates.

To tackle this key challenge, we propose Performance Enhancing Reinforcement Learning Apparatus (PERLA), an enhancement tool for CT-DE based actor-critic MARL algorithms. A key insight underlying the framework is that by utilising shared parameters of the agents' individual policies, the variance of MARL critic estimators can significantly reduced while more accurately representing the true expected returns.  In PERLA, the agents share key parameter information about their policies. This is then used  during training to perform a sampling process of the joint policy. Critically, this allows us to learn much faster than if we only use a single sample of the joint action in the critic function. This dramatically reduces the variance of the critic since the critic estimate now closely approximates the \emph{expected value} under the current joint policy of the agents. Consequently, VF updates are robust against improbable actions observed in single samples. 

The benefits of PERLA can be readily observed in the Coordination Game in Fig. \ref{fig:concept_experiment} where PERLA is applied to MAPPO \citep{yu2021surprising}, a leading MARL algorithm. The line chart in Fig.~\ref{fig:concept_experiment} shows the probability of sampling action `Left' averaged across both agents over $10$ runs against training steps in the Coordination game. As PERLA MAPPO consistently induces lower variance through training, the policy monotonically increases the probability of playing `Left.' As an example, whenever Agent 1 samples the action $r$,  the VF updates calculate the expected value under the policy of Agent 2. The update is towards the expected value due to $(r, \pi_2(l))$ and $(r, \pi_2(r))$ where $\pi_2(l)$ and $\pi_2(r)$ represent the probability of actions $l$ and $r$ respectively by Agent $2$ under its policy. 

In this way, the return of Agent 1 of taking $r$ is computed more accurately, and the VF update lower variance. This enables PERLA to produce consistent convergence of the underlying MARL method MAPPO to the optimal strategy, despite its less likely occurrence (relative to the miscoordinated joint actions) under stochastic policies. On the other hand, since vanilla MAPPO is exposed to random occurrences of miscoordination, it often converges to the sub-optimal stable point due to the penalties of miscoordination. In this example, vanilla MAPPO converges to the optimal NE in $6$ of $10$ runs while PERLA MAPPO converges to the optimal NE in $10$ out of $10$ runs.

\begin{figure}[t]
        \centering
        \begin{subfigure}[t]{0.51\textwidth}
            \centering
            \includegraphics[width=0.65\textwidth]{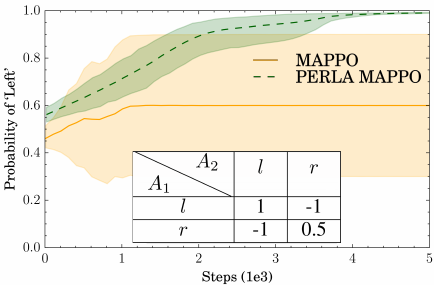}
        \end{subfigure}
        \begin{subfigure}[t]{.215\textwidth}
            \centering
            \includegraphics[width=\textwidth]{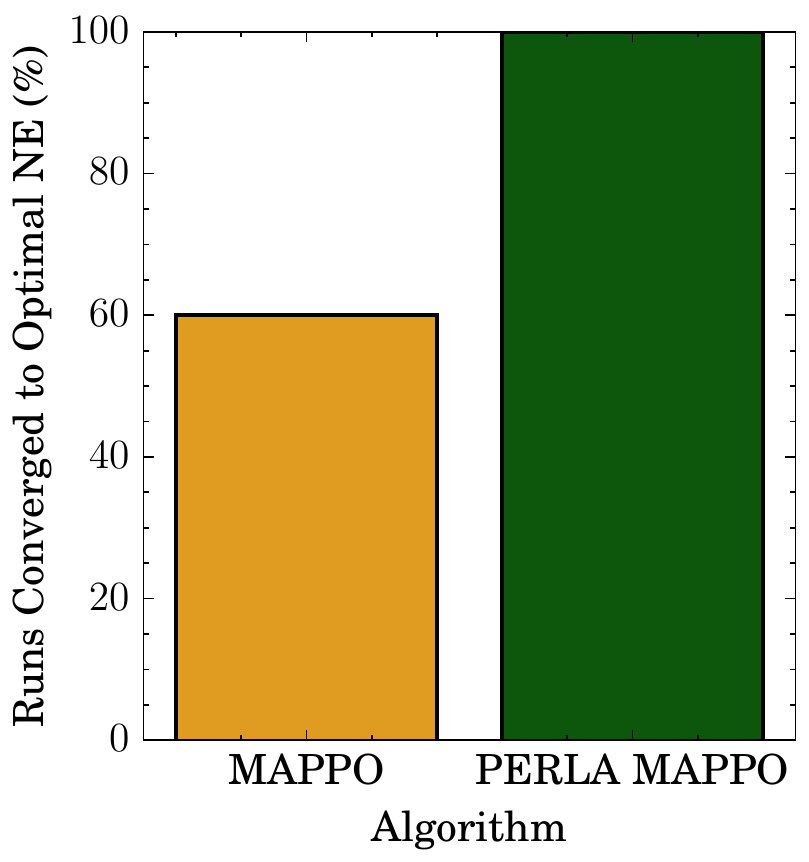}
        \end{subfigure}
        \caption{\emph{Inset table on left}: Payoff matrix of Coordination Game. In this game the actions ($(l, r)$, $(r, 1)$) are penalised. There is also a sub-optimal Nash Equilibrium, $(r, r)$. The joint strategy that maximises the agents' payoff $(l, l)$. \emph{Left}: Policy probability of `Left.' PERLA MAPPO learns the optimal policy and features less variance. \emph{Right}: percentage of runs converged to optimal solution.}
        \label{fig:concept_experiment}\vspace{-0.5 cm}
\end{figure}

We summarise the advantages of PERLA below.
\newline\textbf{1) PERLA improves sample efficiency and convergence properties.} PERLA is easily incorporated into CT-DE based Actor-Critic algorithms, and significantly boosts performance over the base learners (validated empirically on  MAPPO~\citep{yu2021surprising} in Sec.~\ref{sec:performance}). 
\newline\textbf{2) PERLA enables MARL agents to scale efficiently} by factoring the behaviour of other agents when querying the critic and reducing variance. This enables efficient scaling with the number of agents while imposing no restrictive VF constraints (validated empirically in Sec.~\ref{sec:scaling}).
\newline\textbf{3) PERLA has a theoretically sound basis.} We prove PERLA induces a vast reduction of variance of VF estimates (Theorem \ref{th:mcfvar}) and that it preserves policy gradient estimators (Theorem \ref{th:equal_exp}) ensuring the consistency of its solution with the system objective. We also prove  that PERLA applied to actor-critic algorithms converges almost surely to a locally optimal joint policy profile (Theorem \ref{th:convergence}). 

Recent actor-critic methods have shown significant performance improvements over previous MARL algorithms \citep{mguni2021ligs, kuba2021settling, yu2021surprising}, and are state-of-the-art in a range of MARL benchmarks. These actor-critic formalisms are natural candidates for PERLA as we can permit the critic to sample the joint-policy, while leaving the actor (i.e., the policy) unchanged. Thus, PERLA maintains the setup of decentralised execution (and centralised training). 

\section{Related Work}
CT-DE MARL algorithms can be placed within two categories: \emph{value-based} or \emph{actor-critic} methods. In value-based methods, centralised training is assured to generate policies that are consistent with the desired system goal whenever the IGM principle \citep{son2019qtran} is satisfied.\footnote{IGM imposes an equivalence between the joint greedy action and the collection of individual greedy actions.} To realise the IGM principle in CT-DE, QMIX and VDN propose two sufficient conditions of IGM to factorise the joint action-value function. Such decompositions are limited by the joint action-value function class they can represent and can perform badly in systems that do not adhere to these conditions \citep{wang2020qplex}. \footnote{WQMIX \citep{rashid2020weighted} considers a weighted projection towards better performing joint actions but does not guarantee IGM consistency.} Other value-based methods such as QPLEX \citep{wang2020qplex} have been shown to fail in simple tasks with non-monotonic VFs \citep{rashid2020weighted} or in the case of QTRAN \citep{son2019qtran}, scale poorly in complex MARL tasks such as the StarCraft Multi-Agent Challenge (SMAC) \citep{peng2020facmac}.  

On the other hand, actor-critic type methods represent some of the highest performing methods such as MAPPO \citep{yu2021surprising} and are among state-of-the-art. Indeed, recent work by \citet{pmlr-v162-fu22d} has shown that in particular MARL reward structures, actor-critic based CT methods are dominant as the class of algorithms that produce optimal policies. Further, empirical studies \citep{papoudakis2021benchmarking, de2020independent} have shown the strength of actor-critic based CT methods over competing approaches. \cite{foerster2018counterfactual} proposed counter-factual baselines as a method to mitigate variance in MARL actor-critic methods. Their method seeks to accurately assign credit to agents for their contribution to the reward received following execution of a joint-action. \cite{kuba2021settling} propose a general baseline (for the critic) applicable to all MARL actor-critic methods to mitigate variance. Moreover, the authors show that MARL variance may be dis-aggregated into the variance due to the state, the agent's own actions, and the actions of other agents.  Unlike their method, which is limited to mitigating the variance from the agent's own action, we take a step further to mitigate variance due to other agents in the system which is a key impediment for MARL methods especially with larger numbers of agents. 

\section{PERLA Framework}
We formulate the MARL problem as a Markov game (MG) \citep{shoham2008multiagent}
represented by a tuple $\mathfrak{G}=\langle \mathcal{N},\mathcal{S},\left(\mathcal{A}_{i}\right)_{i\in\mathcal{N}} \coloneqq \boldsymbol{\mathcal{A}},P,R_i,\gamma \rangle$. $\mathcal{N} \in \mathbb{N}$ is the number of agents in the system, $\mathcal{S}$ is a finite set of states, $\mathcal{A}_i$ is an action set for agent $i\in\mathcal{N}$, and $R_i:\mathcal{S}\times\boldsymbol{\mathcal{A}}\to\mathcal{P}(D)$ is the reward function that agent $i$ seeks to maximise ($D$ is a compact subset of $\mathbb{R}$), and $P:\mathcal{S} \times \boldsymbol{\mathcal{A}} \times \mathcal{S} \rightarrow [0, 1]$ is the probability function describing the system dynamics 
We consider a partially observable setting, in which given the system state $s^t\in \cS$, each agent $i\in \cN$ makes local observations $\tau^t_i=O(s^t,i)$ where $O:\cS\times \cN \to \cZ_i$ is the observation function and $\cZ_i$ is the set of local observations for agent $i$.  To decide its actions, each agent $i\in\mathcal{N}$ samples its actions from a \textit{Markov policy}
$\pi_{i,\vtheta_i}: \cZ_i \times \mathcal{A}_i \rightarrow [0,1]$, which is parameterised by the vector $\vtheta_i\in\mathbb{R}^d$. Throughout the paper, $\pi_{i,\vtheta_i}$ is abbreviated as $\pi_{i}$.  At each time $t\in 0,1,\ldots,$ the system is in state $s^t\in\mathcal{S}$ and each agent $i\in\mathcal{N}$ takes an action $a_i^t\in\mathcal{A}_i$, which together with the actions of other agents $\va_{-i}^t:=(a_1^t,\ldots,a_{i-1}^t,a_{i+1}^t,\ldots, a_N^t)$,  produces an immediate reward $r_i\sim R(s^t,\va^t_i)$ for agent $i\in\mathcal{N}$. The system then transitions to a next state $s^{t+1}\in\cS$ with probability $P(s^{t+1}|s^t,\boldsymbol{a}^t)$ where $\boldsymbol{a}^t=(a_1^t,\ldots, a_N^t)\in\boldsymbol{\mathcal{A}}$ is the \textit{joint action} which is sampled from the \textit{joint policy} $\boldsymbol{\pi}:=\prod_{i=1}^N\pi_{i}$.  
The goal of each agent $i$ is to maximise its expected returns measured by its VF $v_i(s)=\mathbb{E}\left[\sum_{t=0}^\infty \gamma^tR_i(s^t,\boldsymbol{a}^t)|s^0=s\right]$ and the action-value function for each agent $i\in\cN$ is given by $Q_i(s,\boldsymbol{a})=\mathbb{E}[\sum_{t=0}^\infty R_i(s^t,\boldsymbol{a}^t)|\boldsymbol{a}^0=\boldsymbol{a}] $,
where  $-i$ denotes the tuple of agents excluding agent $i$. Likewise, we denote $\prod_{j=1, j \neq i}^N \pi_j$ as $\vpi_{-i}$. In the fully cooperative case all agents share the same goal: $R_1=\ldots R_N:=R$.

In the CT paradigm, given a state  $s\in\cS$ and  the joint action $\boldsymbol{a}\in\boldsymbol{\cA}$, each agent $i\in\mathcal{N}$ computes its action-value function $Q_i(s,\boldsymbol{a})$.  The action-value function provides an estimate of the agent's expected return using its policy given the behaviour of all other agents $\cN/\{i\}$ for a given action $a_i\in\cA_i$. Therefore, $Q_i(s,\boldsymbol{a})$ seeks to provide an estimate of the agent's own action, accounting for the actions of others. 
Agents use stochastic policies to explore, and therefore the aggregated joint action $\boldsymbol{a}\sim \boldpi$ may be composed of exploratory actions sampled from individual agent's policies.

The core component of PERLA is a sampling process of the joint-policy $\boldsymbol{\pi}_{-i}$ for each agent $i$. This is used to compute the expected value of $a_i$ under joint-policy of the other agents in the system.
We compute the expected value of agent $i$'s action-value function $\tilde{Q}$ , as defined below:
\begin{align} \label{eq:marginalised_q}
         \tilde{Q}_i(s,a_i):=\mathbb{E}_{\boldsymbol{\pi}_{-i}}\left[Q_i(s,\boldsymbol{a})\right]; \quad \boldsymbol{a}\equiv(a_i,\boldsymbol{a}_{-i})\in \boldsymbol{\cA},
\end{align}
where $s\in\cS, a_i\sim \pi_i(\cdot|\tau_i), \boldsymbol{a}_{-i}\sim \boldsymbol{\pi}_{-i}(\cdot|\tau_{-i})$. This object requires some explanation: as with $Q_i$, the function $\tilde{Q}_i$ seeks to estimate the expected return following agent $i$ taking action $a_i$. However, unlike $Q_i$, $\tilde{Q}_i$ builds in the expected value under the  actions of other agents, $\boldsymbol{a}_{-i}$ as well. Consequently, the critic can more accurately estimate the value of action $a_i$ given the behaviour of the other agents in the system. 

In practice it may be impossible to analytically calculate \ref{eq:marginalised_q}, and so to approximate $\tilde{Q}_i(s,a_i)$, for any $\forall s\in\cS$ and any  $a_i\in \cA_i$ we construct:
\begin{align}
\hspace{-4 mm}    \hat{Q}_i(s,a_i)=\frac{1}{k} \sum_{j=1}^k Q_i(s,a_i,\va_i^{(j)})\boldsymbol; \;\;
    \va_{-i}^{(j)} \sim \pi(\va_{-i}|\tau_{-i}). \label{approx_margin_q}
\end{align}

We now give a concrete instantiation of PERLA on the popular MAPPO \citep{yu2021surprising} algorithm. This gives rise to \textbf{PERLA MAPPO} algorithm as shown below in Algorithm \ref{alg:one}. As MAPPO's critic function for each agent only takes $s \in \cS$ and $a_i \in \cA$ as input, we augment the input of the standard MAPPO critic to take $\boldsymbol{a}_{-i}$ as well. We do not make any changes to the standard MAPPO policy, and it continues to only take the agent's local observation as input. As the critic is only needed during CT and not required for execution, PERLA MAPPO operates under the CT-DE paradigm; policies are executed in a fully decentralised manner. In PERLA MAPPO we utilise a value-function style one-step critic, where $Q(s^t, a_i^t, \va_{-i}^{t(j)}) = r_t + \gamma V(s^{t+1}, \va_{-i}^{t+1(j)}) $ and $V(s, \va_{-i})$ which is approximated via a deep neural network (PERLA is fully compatible with different types of critics). In this case, marginalising the behaviour of other agents is equivalent of marginalising the next step value function, as explained in more detail in Appendix \ref{ap:vcritic}. During training, when performing policy and critic updates, we require to generate samples of actions from other agents. In practice, this can either be done by each agent communicating its policy parameters or samples of actions directly. Thus the communication complexity for each round would scale as $\mathcal{O}(\min \{D, KA\})$, where $D$ is the length of policy parameters and $A$ is the size of action space. After the samples are communicated, the approximate expectation of the critic can be performed, which in turn can be used to compute the TD-error. The critic is trained with squared TD error as a loss and the policy is update with the TD-error as the advantage estimate. To perform policy updates, we use PPO \citep{schulman2017proximal}. 

\begin{algorithm}[t]
\caption{PERLA MAPPO}\label{alg:one}
\begin{algorithmic}[1]
 \STATE {\bfseries Input: } Joint-policy $\boldsymbol{\pi}$,
    critic parameters $\boldsymbol{\rho}$,
    policy parameters $\boldsymbol{\theta}$,
    environment $E$,
    number of marginalisation samples $K$
    
 \STATE {\bfseries Output: }Optimised joint-policy $\boldsymbol{\pi}^*$

\STATE Augment MAPPO critic $V_{\boldsymbol{\rho}}$ to as input state $\boldsymbol{s}$ and joint-action $\boldsymbol{a}_{-i}$ 
\STATE Rollout $\boldsymbol{\pi}$ in $E$ to obtain data $D = (\boldsymbol{s}^0, \boldsymbol{a}^1, r^1,\ldots,\boldsymbol{s}^{T-1}, \boldsymbol{a}^T, r^T)$ \;

\FOR{$t = 0$ {\bfseries to} $T-1$}
    \FOR {each agent $i$}
        \STATE Generate $K$ samples of joint-actions at the next-state observations $\{\boldsymbol{a}_{-i}^{t(j)} \sim \boldsymbol{\pi}_{-i}(\boldsymbol{\tau}^{t+1})\}_{j=1}^K$ \;
        
        \STATE Compute TD-error: $\delta_i$ = $\boldsymbol{r}  + \gamma \frac{1}{K}\sum_{j=1}^K V_{\rho}(\boldsymbol{s}^t, \boldsymbol{a}^{t(j)}_{-i}) $  $ - \frac{1}{K}\sum_{j=1}^K V_{\rho}(\boldsymbol{s}^{t}, \boldsymbol{a}^{t(j)}_{-i})
        $ 
        over sampled joint-actions for each agent 
        
        \STATE Update critic parameters $\boldsymbol{\rho}$ with $\delta_i^2$ as the loss\;
        
        \STATE Update $i$th agent's policy parameters $\boldsymbol{\theta}_{i}$ 
        with advantages given by $\delta_i$, using PPO update
    \ENDFOR
    \ENDFOR
\end{algorithmic}
\end{algorithm}

\section{Theoretical Analysis}
We now perform a detailed theoretical analysis of PERLA. Here, we derive theoretical results that establish the key benefits of PERLA, namely that it vastly reduces variance of the key estimates used in training.  We begin with a result that quantifies the reduction of variance when using $\hat{Q}_i$ instead of $Q_i$. We defer all proofs to the Appendix.

\begin{restatable}[]{theorem}{mcfvar}
\label{th:mcfvar}
The variance of marginalised Q-function $\tilde{Q}_i$ is smaller than that of the non-marginalised Q-function $Q_i$ for any $i\in\cN$, that is to say:
$
    \textrm{Var}(Q_i(s, \mathbf{a})) \ge \textrm{Var}(\tilde{Q}_i(s, a_i)) 
$. Moreover, for the approximation to the marginalised Q-function (c.f. Equation \ref{approx_margin_q}) the following relationship holds:
\begin{align}\nonumber
    \textrm{Var}\left(\hat{Q}_i(s, a_{i})\right) = \frac{1}{k} \textrm{Var}\left(Q_i(s,\va_{-i},a_{i})\right) + \frac{k - 1}{k} \textrm{Var}\left(\tilde{Q}_i(s, a_{i})\right).
\end{align}
\end{restatable}

Therefore for $k = 1$ we get that the approximation has the same variance as the non-marginalised Q-function. However, for any $k > 1$ the approximation to marginalised Q-function has smaller variance that the non-marginalised Q-function. Therefore marginalisation procedure of PERLA can essentially be used as a variance reduction technique. Let us now analyse how this framework can be applied to enhance multi-agent policy gradient algorithm, which is known to suffer from high variance in its original version.

In the policy gradient algorithms, we assume a fully cooperative game which avoids the need to add the agent indices to the state-action and state-value functions since the agents have identical rewards. The goal of each agent is therefore to maximise the expected return from the initial state defined as $\mathcal{J}(\vtheta) = \mathbb{E}_{s_0 \sim p(s_0)}[v(s_0)]$, where $p(s_0)$ is the distribution of initial states and $\vtheta = (\vtheta_1^T,\dots,\vtheta_N^T)^T$ is the concatenated vector consisting of policy parameters for all agents. The following well-known theorem establishes the gradient of $\mathcal{J}(\vtheta)$ with respect to the policy parameters.
\begin{theorem}[MARL Policy Gradient \citep{zhang2018fully}] \label{th:mapg}
\begin{equation*}
    \nabla_{\vtheta_i} \mathcal{J}(\vtheta) = \mathbb{E} \left[ \sum_{t=0}^{\infty}\gamma^{t}Q(s^{t}, \va_{-i}^{t}, a_{i}^{t})\nabla_{\vtheta_{i}}\log \pi_{i}(a_{i}^{t}\big|\tau^t_i) \right].
\end{equation*}
\end{theorem}
Therefore, to calculate the gradient with respect to policy parameters, it is necessary to make use of the state-action-values. In practice, it can estimated by a function approximator, which gives rise to Actor-Critic methods \citep{konda1999actor}. In MARL, one can either maintain a centralised critic that provides state-action-value for the $i^{th}$ agent using the knowledge of the other agents' actions or introduce a decentralised critic that does not take actions of others (in its inputs). This gives rise to the centralised training decentralised execution (CT-DE)  $\vg_{i}^{\textrm{C}}$ and decentralised  $\vg_{i}^{\textrm{D}}$ gradient estimators respectively, as defined below. 
\begin{align} \label{eq:ctde_grad_est}
    \vg_{i}^{\textrm{C}}  &:=  \sum_{t=0}^{\infty}\gamma^{t}Q(s^{t}, \va_{-i}^{t}, a_{i}^{t})\nabla_{\vtheta_{i}}\log \pi_{i}(a_{i}^{t}\big|\tau^t_i), \\ 
     \vg_{i}^{\textrm{D}}  &:= \sum_{t=0}^{\infty}\gamma^{t}\tilde{Q}(s^{t}, a_{i}^{t})\nabla_{\vtheta_{i}}\log \pi_{i}(a_{i}^{t}\big|\tau^t_i),\label{eq:dt_grad_est}
\end{align}
where $Q(s_{t}, a_{-i}^{t}, a_{i}^{t})$ and $\tilde{Q}(s^{t}, a_{i}^{t})$ are the CT-DE and decentralised critics respectively. In PERLA Actor-Critic algorithms (such as in Algorithm \ref{alg:one}) we use a third estimator - the PERLA estimator:
\begin{equation}  \label{eq:perla_grad_est}
    \vg_{i}^{\textrm{P}}  :=  \sum_{t=0}^{\infty}\gamma^{t}\hat{Q}(s^{t}, a_{i}^{t} )\nabla_{\vtheta_{i}}\log \pi_{i}(a_{i}^{t}\big|\tau^t_i), 
\end{equation}
where $\hat{Q}(s^{t}, a_{i}^{t})$ is the Monte-Carlo approximation of $\tilde{Q}(s^{t}, a_{i}^{t})$. Note, in this approach we maintain a centralised critic $Q(s^t, a^t_{-i},\va^t_i)$, therefore the approximation to the marginalised Q-function is obtained using $\hat{Q}(s^t,a^t_{i}) = \frac{1}{k} \sum_{j=1}^k Q(s^t,a^t_i ,\va^{t(j)}_{-i})$, where $\va^{t(j)}_{-i} \sim \pi_{-i}(\va^t_{-i}|\tau^t_{-i})$. The PERLA estimator is equal in expectation to the CT-DE estimator as stated by the following Theorem.
\begin{restatable}[]{theorem}{equalexp} \label{th:equal_exp}
Given the same (possibly imperfect) critic, the estimators $\vg_i^C$ and $\vg_i^P$ have the same expectation, that is:
\begin{equation*}
    \mathbb{E}[\vg_{i}^{\textrm{C}}]  = \mathbb{E}[\vg_{i}^{\textrm{P}}].
\end{equation*}
\end{restatable}
Hence, whenever the critic provides the true Q-value, the PERLA estimator is an unbiased estimate of the policy gradient (which follows from Theorem \ref{th:mapg}). However, although the CT-DE and PERLA estimators have the same expectations, the PERLA estimator enjoys significantly lower variance. As in \citep{kuba2021settling}, we analyse the excess variance the two estimators have over the decentralised estimator. First define by $B_i$  the upper bound on the gradient norm of $i$th agent, i.e. $B_{i} = \sup_{s, \va}\left|\left|\nabla_{\theta^{i}}\log\pi_{i}\left( a_{i}|s \right) \right|\right|$ and by $C $ the upper bound on the Q-function, i.e. $C = \sup_{s,\va} Q(s, \va)$.
We now present two theorems showing the effectiveness of PERLA estimator for policy gradients.
\begin{restatable}[]{theorem}{sharingdtbound} \label{th:sharingtdbound}
Given true Q-values, the difference in variances between the decentralised and PERLA estimators admits the following bound:
\begin{equation*}
    \textrm{Var}(\vg_{i}^{\textrm{P}}) - \textrm{Var}(\vg_{i}^{\textrm{D}}) \le  \frac{1}{k} \frac{B_i^2 C^2}{1-\gamma^2} .
\end{equation*}
\end{restatable}

\begin{restatable}[]{theorem}{ctdedtbound} \label{th:ctdedtbound}
Given true Q-values, the difference in variances between the decentralised and CT-DE estimators admits the following bound:
\begin{equation*}
    \textrm{Var}(\vg_{i}^{\textrm{C}}) - \textrm{Var}(\vg_{i}^{\textrm{D}}) \le \frac{B_i^2 C^2}{1-\gamma^2}. 
\end{equation*}
\end{restatable}
Therefore, we can see that with $k=1$ the bound on excess variance of the PERLA estimator is the same as for the CT-DE estimator, but as $k \to \infty$, the variance of PERLA estimator matches the one of the fully decentralised estimator. However, this is done while still maintaining a centralised critic, unlike in the fully decentralised case. It turns out that the presence of centralised critic is critical, as it allows us to guarantee the converge of an algorithm to a local optimum with probability 1. The result is stated by the next Theorem.

\begin{restatable}[]{theorem}{convergence} \label{th:convergence}
Under the standard assumptions of stochastic approximation theory \citep{konda1999actor}, an Actor-Critic algorithm using $\vg_i^P$ or $\vg_i^C$ as a policy gradient estimator, converges to a local optimum with probability 1, i.e. 
\begin{equation*}
   P \left( \lim_{k \to \infty} \lVert \nabla_{\vtheta_i}\mathcal{J}(\vtheta^k) \rVert = 0 \right) = 1,
\end{equation*}
where $\vtheta^k$ is the value of vector $\vtheta$ obtained after the $k$th update following the policy gradient.
\end{restatable}
\begin{proofsketch}
We present a proof sketch here and defer the full proof to Appendix \ref{ap:convergence}. \\
The proof consists of showing that a multi-agent Actor-Critic algorithm using policy gradient estimate $\vg_i^P$ or $\vg_i^C$ is essentially a special case of single-agent Actor-Critic.
\end{proofsketch}
Note that because the decentralised critic does not allow us to query the state-action-value for joint action of all agents, a decentralised actor-critic using $\vg_i^D$ as policy gradient estimate is not equivalent to the single-agent version and we cannot establish convergence for it. Therefore,
our PERLA estimator enjoys both the low variance property of the decentralised estimator and the convergence guarantee of the CT-DE one. Additionally, having a centralised critic yields better performance in practice in environments with strong interactions between agents \citep{kok2004sparse}.

\section{Experiments}\label{section:experiments}
We ran a series of experiments in \textit{Large-scale Matrix Games} \citep{son2019qtran}, \textit{Level-based Foraging} (LBF) \citep{christianos2020shared}, \textit{Multi-agent Mujoco} \citep{de2020deep} and the \textit{StarCraft II Multi-agent Challenge} (SMAC) \citep{samvelyan2019starcraft}\footnote{The specific maps/variants used of each of these environments in given in Sections \ref{sec:lbf_experiments_appendix} and \ref{section:smac_details_appendix} of the Appendix)} to test if PERLA: \textbf{1.} Improves overall performance of MARL learners. \textbf{2.} Enables  sample efficiency when the the number of agents is scaled up. \textbf{3.} Reduces variance of value function estimates. In all tasks, we compared the performance of PERLA MAPPO against MAPPO \citep{yu2021surprising}. We report average training results across multiple scenarios/maps in LBF and SMAC. Detailed performance comparisons are deferred to the Appendix. Lastly, we ran a suite of ablation studies which we deferred to the Appendix.

We implemented PERLA on top of the MAPPO implementation provided in the codebase accompanying the MARL benchmark study of \citet{papoudakis2021benchmarking}. Hyperparameters were tuned using simple grid-search, the values over which we tuned the hyperparameters are presented in Table \ref{tab:hyperparams} in the Appendix. All results are means over $3$ random seeds unless otherwise stated. In plots dark lines represent the mean across the seeds while shaded areas represent 95\% confidence intervals.

\subsection{Performance Analysis}\label{sec:performance}
\begin{figure}[t]
    \centering
    \includegraphics[width=0.31\linewidth]{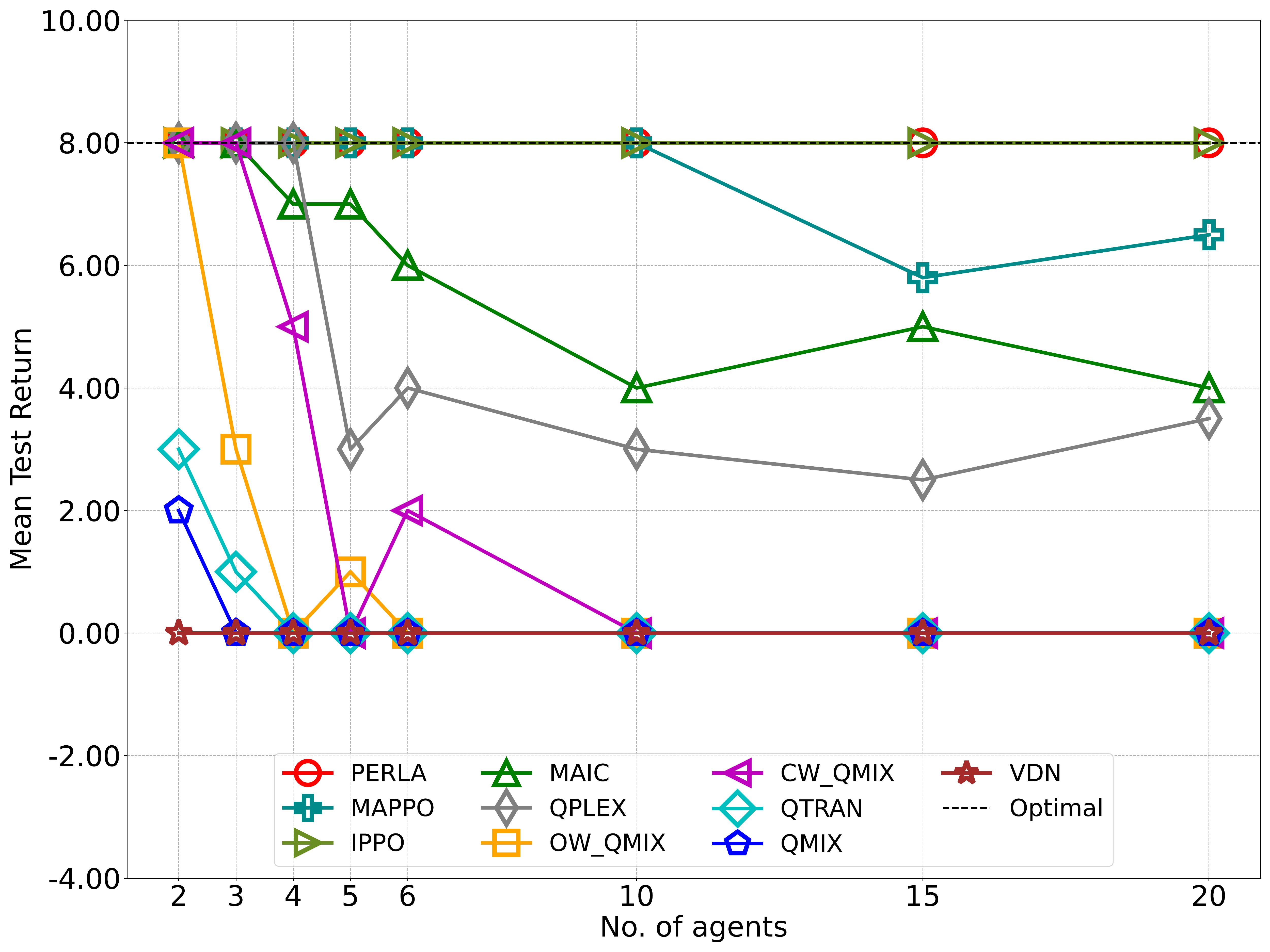}
    \includegraphics[width=0.31\linewidth]{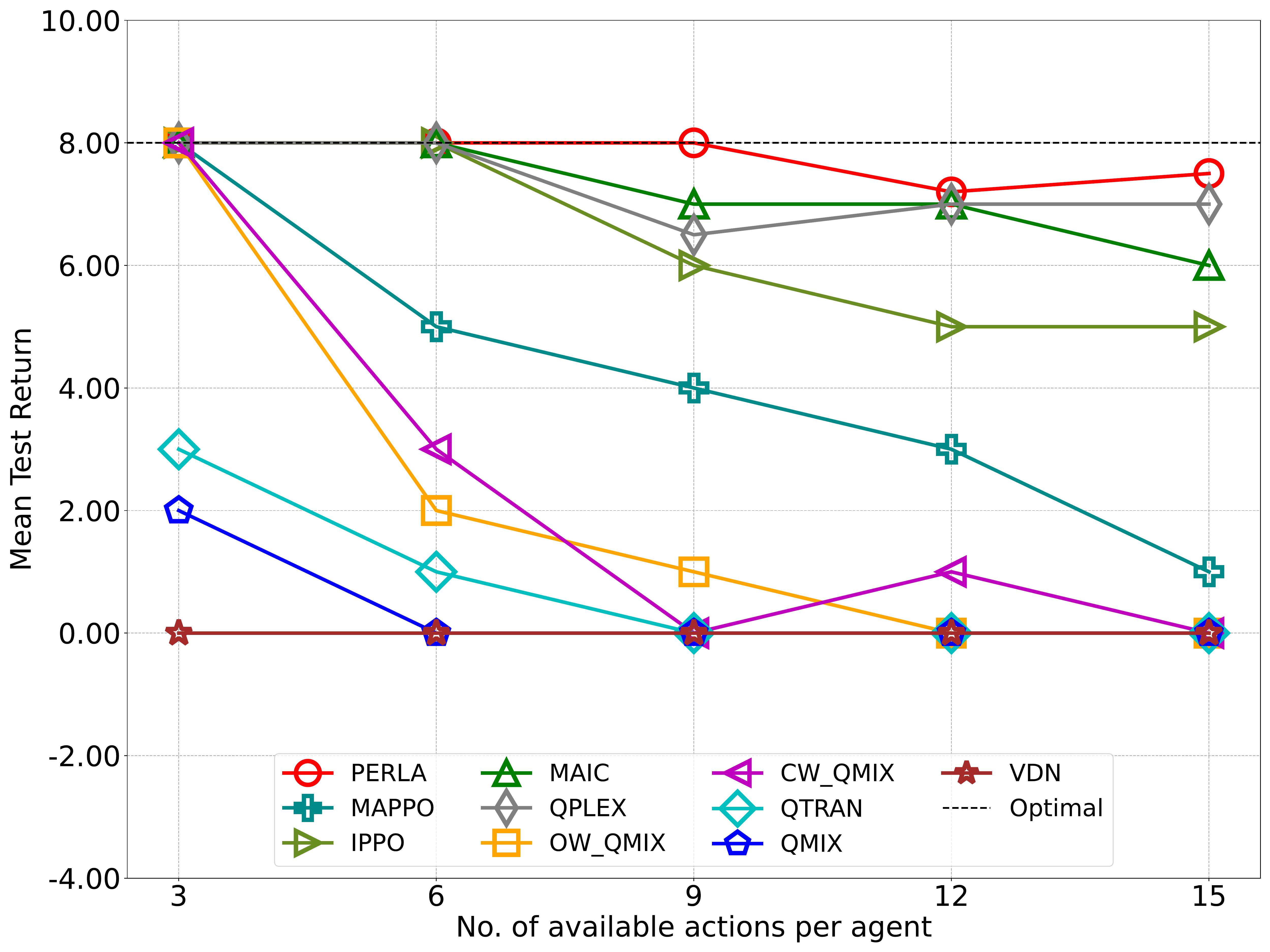}
    \includegraphics[width=0.3375\textwidth]{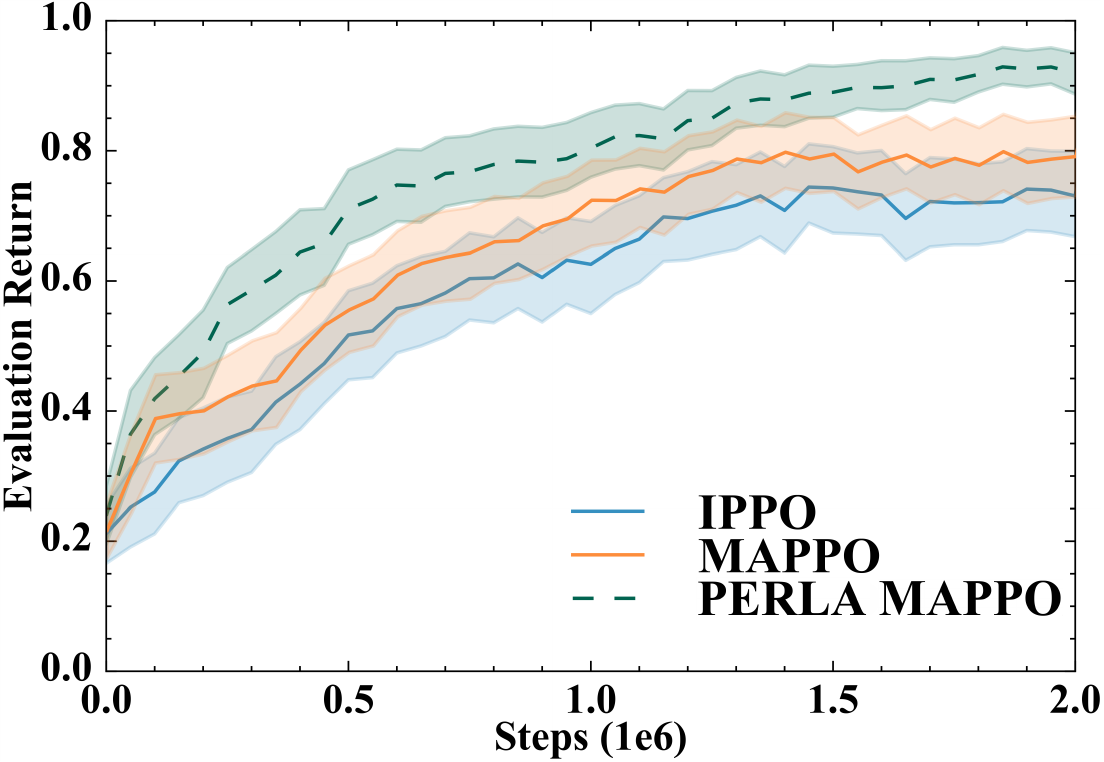}
    \caption{\emph{Left} and \emph{centre}: scaling with more agents and larger action space, respectively, in Cooperative Matrix Game; in both cases PERLA is able to maintain optimal performance while other algorithms' performances degrade. \emph{Right:} Learning curves of mean evaluation return of MAPPO and PERLA MAPPO averaged over all tested LBF maps. PERLA improves sample efficiency (better performance faster) and quality of the final policy.}
    \label{figure:performance_matrix_games_and_lbf}
\end{figure}

\textbf{\textit{Large-Scale Matrix Games.}} To demonstrate PERLA's ability to handle various reward structures and scale efficiently we first tested its performance in a set of variants of the hard matrix game proposed in  \citep{son2019qtran} (Appendix \ref{section:son_matrix_game}). This game contains multiple stable points and a strongly attracting equilibrium \citep{son2019qtran} (all agents selecting action $A$). In our variants, we scaled up the number of agents, or we increased the size of the action space (Appendix \ref{section:son_matrix_game}). In both cases, it is crucial to accurately account for the behaviour of other agents in the system. For instance, even if the joint-policy has converged to the optimal solution, if even one agent samples an exploratory action, it can be strongly destabilising due to its high penalty (reward of $0$ or $-12$ as opposed reward of $8$ for optimal joint-action). To avoid these issues, it is crucial to base updates on the joint-policy rather than samples of the joint-action.

We used $500$ training iterations and average the results of 10 random seeds in each method. As shown in Figure ~\ref{figure:performance_matrix_games_and_lbf}, policy-based methods (MAPPO and IPPO) and leading value-based methods (MAIC, QPLEX \citep{wang2020qplex} and WQMIX \citep{rashid2020weighted}) achieve optimal performance in the initial settings ($2$ agents with $3$ available actions), while other algorithms achieve suboptimal outcomes. When scaling with more agents and larger action space, only PERLA can maintain optimal performance across almost all variations. As shown, MAPPO completely fails when we scale the actions-space, going from a return of 8 to 0.5. PERLA MAPPO, however is robust and attains the highest return of all tested algorithms at about 7.75.

\textit{\textbf{Level-based Foraging.}}
Figure~\ref{figure:performance_matrix_games_and_lbf} shows learning curves averaged across all LBF maps that we ran (the full list of maps is given in the Appendix \ref{sec:lbf_experiments_appendix}). As shown in the plot, PERLA significantly improves base MAPPO, both in learning speed and the quality of the policy at the end of training. For example, it takes PERLA MAPPO about $800,000$ interactions with the environment to achieve a mean evaluation return of $0.8$, whereas vanilla MAPPO does not achieve such a perfomance level even by the end of training. Moreover, PERLA MAPPO is able to attain an mean evaluation return of $0.9$ by the end of training, vanilla MAPPO attains an mean evaluation return just under $0.8$. Furthermore, as shown in Figure \ref{fig:LBF_IPPO} (Appendix \ref{sec:lbf_experiments_appendix})) in the map requiring the highest level of coordination between agents \verb|Foraging-15x15-8p-1f-coop-v2| (8 agents must cooperate to attain reward), PERLA MAPPO manages to learn, and achieves an evaluation return of $0.7$ while MAPPO fails entirely. This shows that PERLA enables vanilla MAPPO to scale to a high number of agents in an environment with sparse and high variance reward function.

\textit{\textbf{StarCraft II Multi-agent Challenge.}}
\begin{figure*}[!t]
    \begin{center}
    \includegraphics[width=\linewidth]{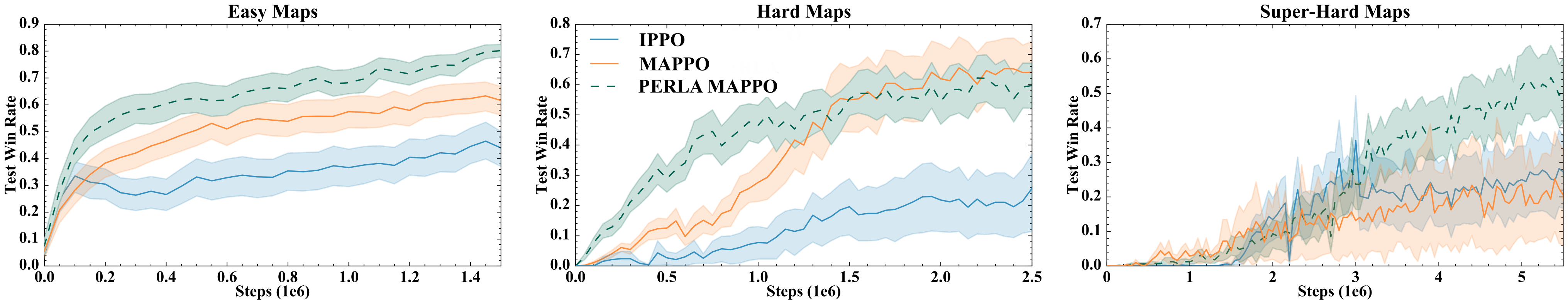}
    \end{center}
    \caption{Average performance over \emph{all} SMAC maps for IPPO, PERLA\_IPPO (top); MAPPO,  PERLA\_MAPPO (bottom). PERLA enhances the base method both in terms of sample efficiency (better performance, faster) and final performance.}
    \label{fig:smac_performance_summary}
    \end{figure*}
Figure \ref{fig:smac_performance_summary} shows performance of MAPPO and PERLA MAPPO over a wide range of SMAC maps from all difficulty levels. At regular intervals during training, we ran $10$ evaluation episodes and tracked the median win-rate. The learning curves are then generated by computing the mean of these win rates (disaggregated curves for each map are available in Figure ~\ref{figure:smac_all_maps_learning_curves} in Appendix ~\ref{section:smac_all_maps}). SMAC maps are richly diverse and vary along several dimensions such as the number of agents to control, density of environment reward, degree of coordination required, and (partial)-observability. Therefore, aggregated and averaged performance over all maps gives us a fairly robust understanding of the effects of PERLA.
As with LBF, PERLA enhances the performance MAPPO. PERLA MAPPO is more sample efficient and converges to better overall policies (within the training budget). 

\begin{figure}[t]
    \centering
    \includegraphics[width=0.32\textwidth]{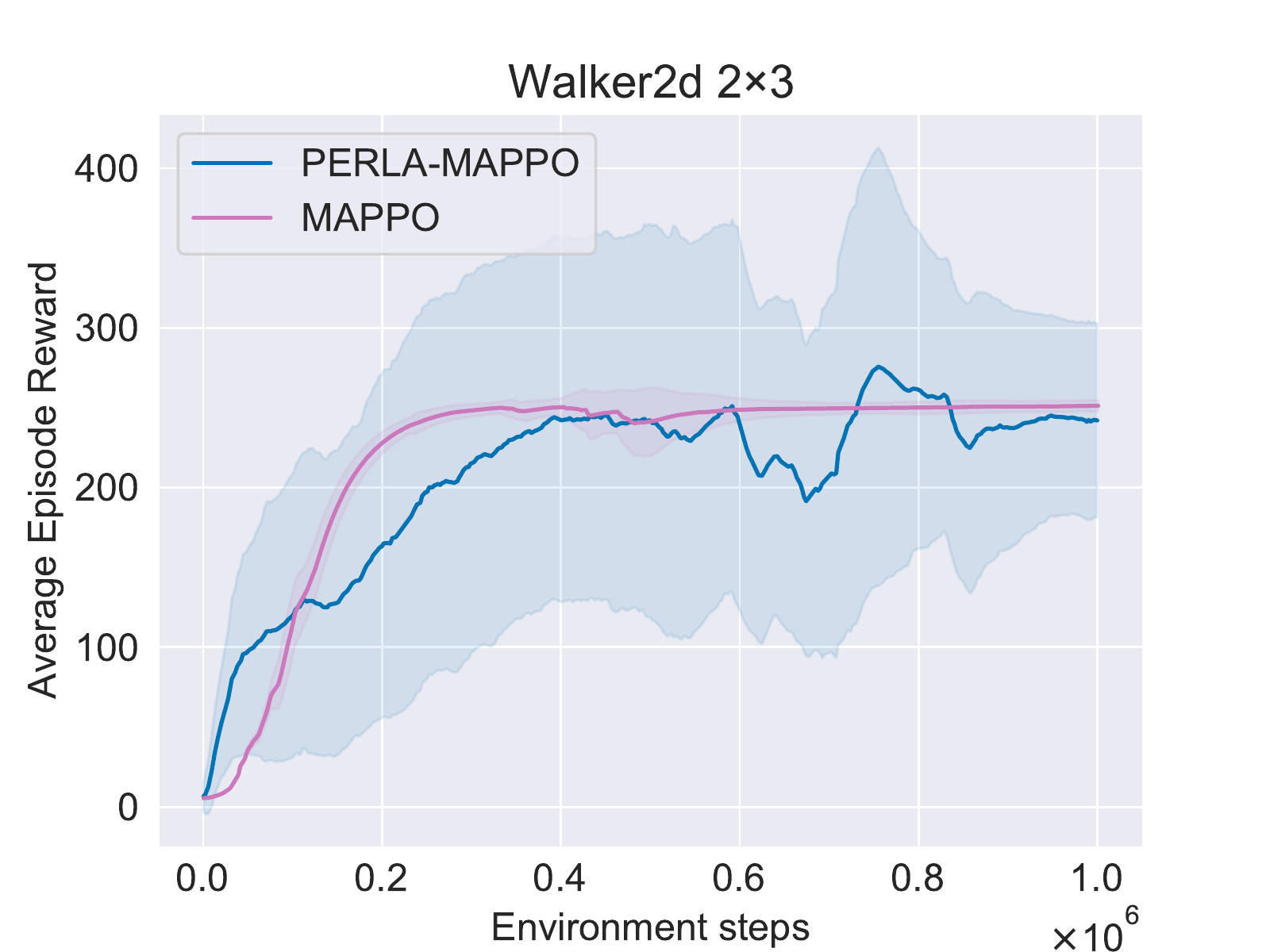}
    \includegraphics[width=0.32\textwidth]{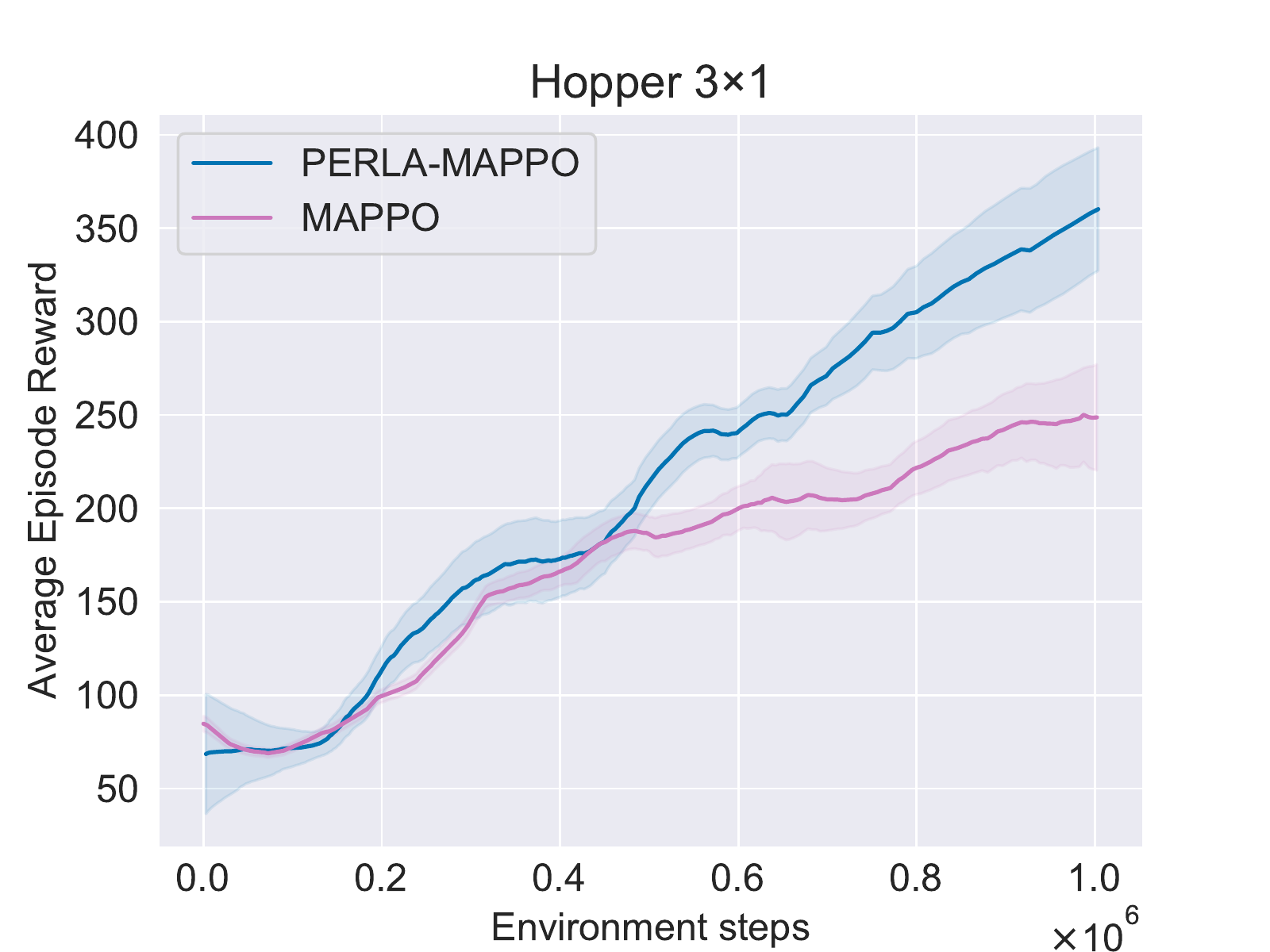}
    \includegraphics[width=0.32\textwidth]{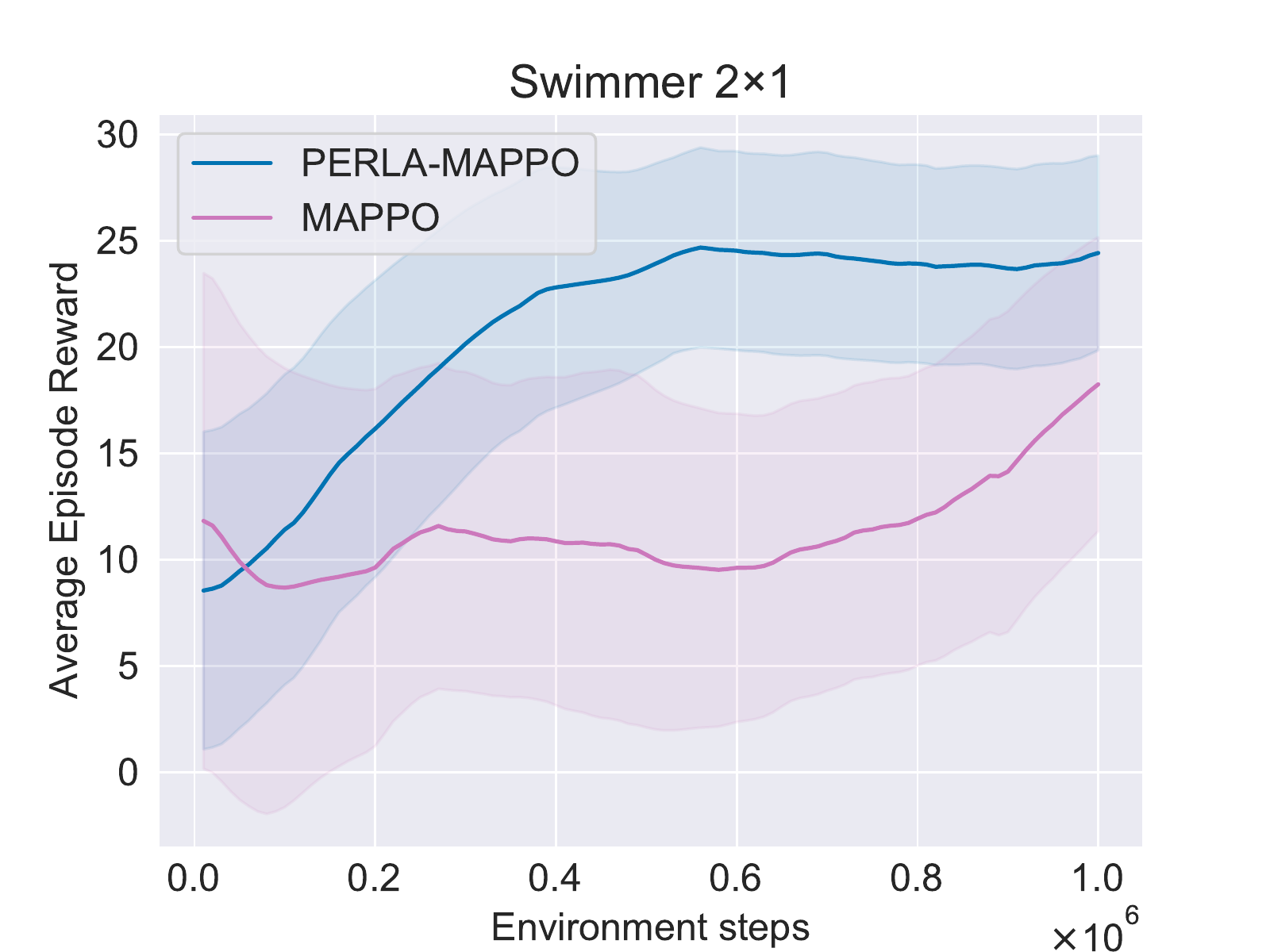} 
    \caption{Comparisons of PERLA MAPPO versus MAPPO on three Multi-agent Mujoco tasks. PERLA MAPPO learns faster and converges to superior policies.}
    \label{alg-Performance}
\end{figure}
\textit{\textbf{Multi-agent Mujoco.}}
To study PERLA's capabilities in complex settings that require both scalability and coordination, we compared its performance with vanilla MAPPO on three tasks in Multi-agent Mujoco: Walker 2$\times$3, Hopper 3$\times$1, and Swimmer 2$\times$1.
In Figure \ref{alg-Performance}, we report learning curves averaged over $6$ seeds of each algorithm. As can be seen, PERLA MAPPO outperforms or equals MAPPO on all three tasks. By enabling agents to maintain estimates that account for other agents’ actions, PERLA achieves more accurate value estimation with the variance reduction, therefore establishing more efficient learning. 

\begin{figure}[t]
    \centering
    \includegraphics[width=0.25\linewidth]{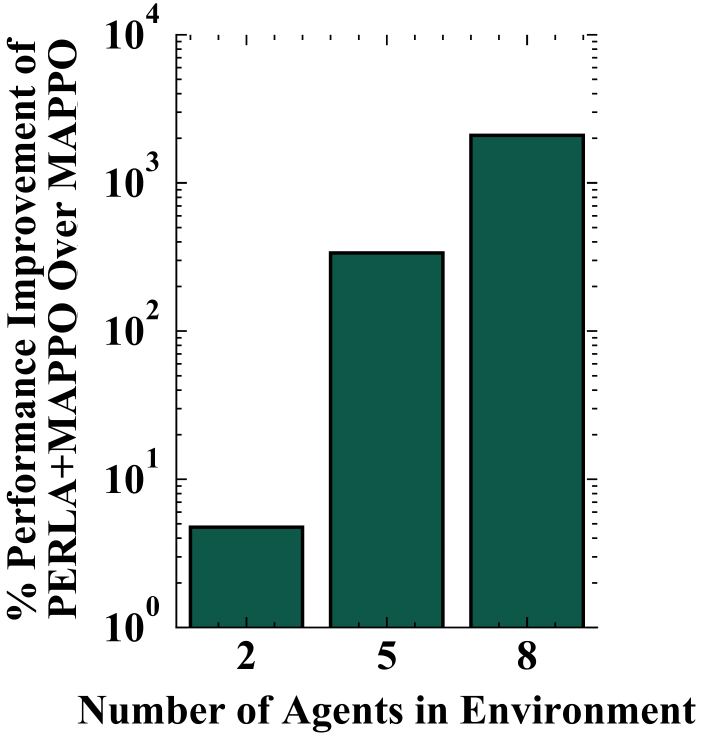}
    \includegraphics[width=0.34\linewidth]{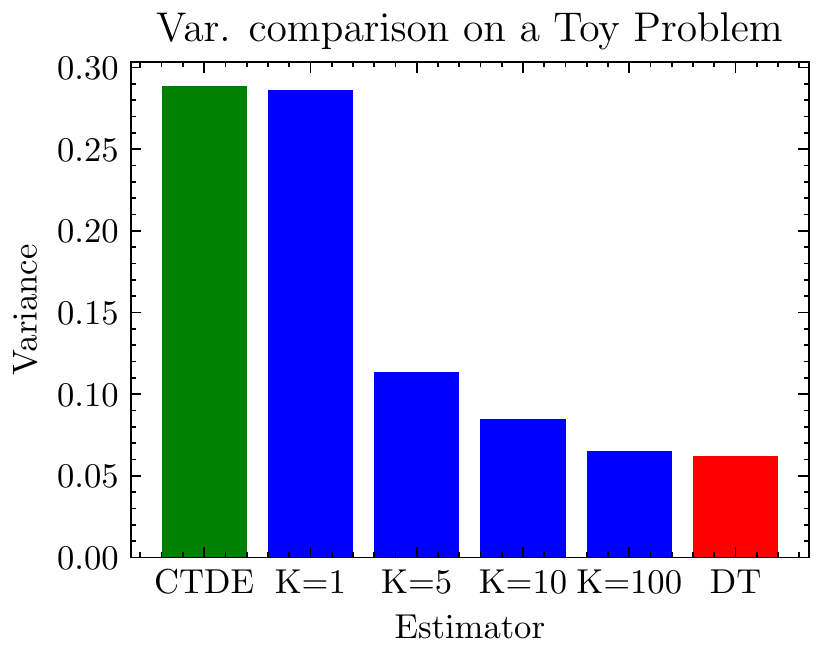}
    \caption{\textbf{Left: } \% performance increase of policy trained by PERLA MAPPO over policy trained by MAPPO as a function number of agents ($N$). For $N=8$, PERLA MAPPO yields over $+1000\%$ performance gains over MAPPO. \textbf{Right:} Comparison of Policy Gradient Variance on a Toy Problem for CTDE, DT and PERLA estimators. $K=n$ denotes a PERLA estimator using $n$ samples.} 
    \label{figure:analysis_experiments}
\end{figure}

     

\subsection{Scaling Analysis}\label{sec:scaling}
Scaling efficiently to large systems  (i.e. systems with many agents and large action spaces) is a major challenge in MARL.
In Section 1, we claimed the PERLA framework enables MARL to scale efficiently in terms of the number of samples.  To test this claim, we investigated PERLA's scaling ability in our large scale matrix games (described above) and LBF. In matrix games, we demonstrated PERLA's ability to efficiently scale across both dimensions, namely, games varied by (1) the number of agents $N=2,3,4,5,6,10,15,20$ and (ii) the cardinality of the agents' action sets $|\cA_i|=3,6,9,12,15$. In each case, we retained the setup that reward agents with a score of $8$ only when all agents choose the first action. 

We further tested this claim in LBF scenarios with 2, 5, and 8 agent, respectively. We expected to see PERLA MAPPO's advantage increase with the number of agents as the algorithm allows each agent to better account for the actions of other agents in the system. Figure \ref{figure:analysis_experiments} shows PERLA MAPPO's over MAPPO. As shown, PERLA enables monotonic performance gains with the number of agents, yielding over $1000\%$ improvement in systems with $8$ agents.

\subsection{Variance Analysis}
In Sec. 4, we proved that PERLA reduces the variance of VF estimators. To show this empirically,  we constructed a toy problem with three agents. Each agent has a binary action space of $\{0, 1\}$ and plays a uniform policy. The team receives a reward of $1$ if all players play action $0$ and a reward of $3$ if all players play action $1$. If at least one player plays an action different from others, the whole team receives a reward of $0$. As this is a state-less game, the $Q$-function takes only actions as inputs. We consider the case, where the policy of agent 1 is defined by a sigmoid function i.e. $\pi(a_1=1; \theta) = \frac{1}{1 + \exp(- \theta)}$. We repeat the experiment 1000 times and measure the variance of the policy gradient for $\theta=0$ calculated using PERLA (Eq. \ref{eq:perla_grad_est}), CTDE (Eq. \ref{eq:ctde_grad_est}) and DT (Eq. \ref{eq:dt_grad_est}) estimators. In the results shown in Figure \ref{figure:analysis_experiments} we see that as number of samples $k$ increases, the variance of PERLA estimator sharply decreases and almost matches the variance of DT estimator for very large $k$. This is consistent with Theorem \ref{th:sharingtdbound}, where we have proven that as $k$ increases variance of PERLA estimator approaches the variance of DT estimator at a rate of $1/k$. 


\section{Conclusion}
Centralised training is a fundamental paradigm in performant, modern actor-critic MARL algorithms. While it enables agents to coordinate to solve challenging problems, MARL estimators still suffer from high variance. This hinders learning and reduces the sample efficiency of MARL methods. Scalability and efficient learning are key challenges in MARL research. 
%
In this paper, we introduced PERLA, an enhancement tool that induces sample efficient, coordinated learning among MARL agents. Our theory and empirical analyses show that PERLA reduces the variance of VF estimators which is critical for efficient learning. In this way, PERLA enables MARL algorithms to exhibit sample efficient learning and a high degree of scalability.


\section{Acknowledgement}
    Jianhong Wang is fully supported by UKRI Turing AI World-Leading Researcher Fellowship, EP/W002973/1.

\bibliographystyle{plainnat}
\bibliography{perla_neurips_2023}

\begin{thebibliography}{26}
\providecommand{\natexlab}[1]{#1}
\providecommand{\url}[1]{\texttt{#1}}
\expandafter\ifx\csname urlstyle\endcsname\relax
  \providecommand{\doi}[1]{doi: #1}\else
  \providecommand{\doi}{doi: \begingroup \urlstyle{rm}\Url}\fi

\bibitem[Christianos et~al.(2020)Christianos, Schäfer, and
  Albrecht]{christianos2020shared}
Filippos Christianos, Lukas Schäfer, and Stefano~V Albrecht.
\newblock Shared experience actor-critic for multi-agent reinforcement
  learning.
\newblock In \emph{Advances in Neural Information Processing Systems
  (NeurIPS)}, 2020.

\bibitem[de~Witt et~al.(2020{\natexlab{a}})de~Witt, Gupta, Makoviichuk,
  Makoviychuk, Torr, Sun, and Whiteson]{de2020independent}
Christian~Schroeder de~Witt, Tarun Gupta, Denys Makoviichuk, Viktor
  Makoviychuk, Philip~HS Torr, Mingfei Sun, and Shimon Whiteson.
\newblock Is independent learning all you need in the starcraft multi-agent
  challenge?
\newblock \emph{arXiv preprint arXiv:2011.09533}, 2020{\natexlab{a}}.

\bibitem[de~Witt et~al.(2020{\natexlab{b}})de~Witt, Peng, Kamienny, Torr,
  B{\"o}hmer, and Whiteson]{de2020deep}
Christian~Schroeder de~Witt, Bei Peng, Pierre-Alexandre Kamienny, Philip Torr,
  Wendelin B{\"o}hmer, and Shimon Whiteson.
\newblock Deep multi-agent reinforcement learning for decentralized continuous
  cooperative control.
\newblock \emph{arXiv preprint arXiv:2003.06709}, 2020{\natexlab{b}}.

\bibitem[Foerster et~al.(2018)Foerster, Farquhar, Afouras, Nardelli, and
  Whiteson]{foerster2018counterfactual}
Jakob~N Foerster, Gregory Farquhar, Triantafyllos Afouras, Nantas Nardelli, and
  Shimon Whiteson.
\newblock Counterfactual multi-agent policy gradients.
\newblock In \emph{Thirty-second AAAI conference on artificial intelligence},
  2018.

\bibitem[Fu et~al.(2022)Fu, Yu, Xu, Yang, and Wu]{pmlr-v162-fu22d}
Wei Fu, Chao Yu, Zelai Xu, Jiaqi Yang, and Yi~Wu.
\newblock Revisiting some common practices in cooperative multi-agent
  reinforcement learning.
\newblock In Kamalika Chaudhuri, Stefanie Jegelka, Le~Song, Csaba Szepesvari,
  Gang Niu, and Sivan Sabato, editors, \emph{Proceedings of the 39th
  International Conference on Machine Learning}, volume 162 of
  \emph{Proceedings of Machine Learning Research}, pages 6863--6877. PMLR,
  17--23 Jul 2022.
\newblock URL \url{https://proceedings.mlr.press/v162/fu22d.html}.

\bibitem[Gu et~al.(2016)Gu, Lillicrap, Ghahramani, Turner, and Levine]{gu2016q}
Shixiang Gu, Timothy Lillicrap, Zoubin Ghahramani, Richard~E Turner, and Sergey
  Levine.
\newblock Q-prop: Sample-efficient policy gradient with an off-policy critic.
\newblock \emph{arXiv preprint arXiv:1611.02247}, 2016.

\bibitem[Kok and Vlassis(2004)]{kok2004sparse}
Jelle~R Kok and Nikos Vlassis.
\newblock Sparse cooperative q-learning.
\newblock In \emph{Proceedings of the twenty-first international conference on
  Machine learning}, page~61, 2004.

\bibitem[Konda and Tsitsiklis(1999)]{konda1999actor}
Vijay Konda and John Tsitsiklis.
\newblock Actor-critic algorithms.
\newblock \emph{Advances in neural information processing systems}, 12, 1999.

\bibitem[Kuba et~al.(2021)Kuba, Wen, Meng, Zhang, Mguni, Wang, Yang,
  et~al.]{kuba2021settling}
Jakub~Grudzien Kuba, Muning Wen, Linghui Meng, Haifeng Zhang, David Mguni, Jun
  Wang, Yaodong Yang, et~al.
\newblock Settling the variance of multi-agent policy gradients.
\newblock \emph{Advances in Neural Information Processing Systems},
  34:\penalty0 13458--13470, 2021.

\bibitem[Mguni et~al.(2018)Mguni, Jennings, and
  de~Cote]{mguni2018decentralised}
David Mguni, Joel Jennings, and Enrique~Munoz de~Cote.
\newblock Decentralised learning in systems with many, many strategic agents.
\newblock In \emph{Thirty-Second AAAI Conference on Artificial Intelligence},
  2018.

\bibitem[Mguni et~al.(2021{\natexlab{a}})Mguni, Wu, Du, Yang, Wang, Li, Wen,
  Jennings, and Wang]{mguni2021learning}
David Mguni, Yutong Wu, Yali Du, Yaodong Yang, Ziyi Wang, Minne Li, Ying Wen,
  Joel Jennings, and Jun Wang.
\newblock Learning in nonzero-sum stochastic games with potentials.
\newblock \emph{arXiv preprint arXiv:2103.09284}, 2021{\natexlab{a}}.

\bibitem[Mguni et~al.(2021{\natexlab{b}})Mguni, Jafferjee, Wang, Perez-Nieves,
  Slumbers, Tong, Li, Zhu, Yang, and Wang]{mguni2021ligs}
David~Henry Mguni, Taher Jafferjee, Jianhong Wang, Nicolas Perez-Nieves, Oliver
  Slumbers, Feifei Tong, Yang Li, Jiangcheng Zhu, Yaodong Yang, and Jun Wang.
\newblock Ligs: Learnable intrinsic-reward generation selection for multi-agent
  learning.
\newblock \emph{arXiv preprint arXiv:2112.02618}, 2021{\natexlab{b}}.

\bibitem[Papoudakis et~al.(2021)Papoudakis, Christianos, Schäfer, and
  Albrecht]{papoudakis2021benchmarking}
Georgios Papoudakis, Filippos Christianos, Lukas Schäfer, and Stefano~V.
  Albrecht.
\newblock Benchmarking multi-agent deep reinforcement learning algorithms in
  cooperative tasks.
\newblock In \emph{Proceedings of the Neural Information Processing Systems
  Track on Datasets and Benchmarks (NeurIPS)}, 2021.
\newblock URL \url{http://arxiv.org/abs/2006.07869}.

\bibitem[Peng et~al.(2020)Peng, Rashid, de~Witt, Kamienny, Torr, B{\"o}hmer,
  and Whiteson]{peng2020facmac}
Bei Peng, Tabish Rashid, Christian A~Schroeder de~Witt, Pierre-Alexandre
  Kamienny, Philip~HS Torr, Wendelin B{\"o}hmer, and Shimon Whiteson.
\newblock Facmac: Factored multi-agent centralised policy gradients.
\newblock \emph{arXiv preprint arXiv:2003.06709}, 2020.

\bibitem[Peng et~al.(2017)Peng, Wen, Yang, Yuan, Tang, Long, and
  Wang]{peng2017multiagent}
Peng Peng, Ying Wen, Yaodong Yang, Quan Yuan, Zhenkun Tang, Haitao Long, and
  Jun Wang.
\newblock Multiagent bidirectionally-coordinated nets: Emergence of human-level
  coordination in learning to play starcraft combat games.
\newblock \emph{arXiv preprint arXiv:1703.10069}, 2017.

\bibitem[Rashid et~al.(2018)Rashid, Samvelyan, Schroeder, Farquhar, Foerster,
  and Whiteson]{rashid2018qmix}
Tabish Rashid, Mikayel Samvelyan, Christian Schroeder, Gregory Farquhar, Jakob
  Foerster, and Shimon Whiteson.
\newblock Qmix: Monotonic value function factorisation for deep multi-agent
  reinforcement learning.
\newblock In \emph{International Conference on Machine Learning}, pages
  4295--4304. PMLR, 2018.

\bibitem[Rashid et~al.(2020)Rashid, Farquhar, Peng, and
  Whiteson]{rashid2020weighted}
Tabish Rashid, Gregory Farquhar, Bei Peng, and Shimon Whiteson.
\newblock Weighted qmix: Expanding monotonic value function factorisation for
  deep multi-agent reinforcement learning.
\newblock \emph{arXiv preprint arXiv:2006.10800}, 2020.

\bibitem[Samvelyan et~al.(2019)Samvelyan, Rashid, De~Witt, Farquhar, Nardelli,
  Rudner, Hung, Torr, Foerster, and Whiteson]{samvelyan2019starcraft}
Mikayel Samvelyan, Tabish Rashid, Christian~Schroeder De~Witt, Gregory
  Farquhar, Nantas Nardelli, Tim~GJ Rudner, Chia-Man Hung, Philip~HS Torr,
  Jakob Foerster, and Shimon Whiteson.
\newblock The starcraft multi-agent challenge.
\newblock \emph{arXiv preprint arXiv:1902.04043}, 2019.

\bibitem[Schulman et~al.(2017)Schulman, Wolski, Dhariwal, Radford, and
  Klimov]{schulman2017proximal}
John Schulman, Filip Wolski, Prafulla Dhariwal, Alec Radford, and Oleg Klimov.
\newblock Proximal policy optimization algorithms.
\newblock \emph{CoRR}, abs/1707.06347, 2017.

\bibitem[Shoham and Leyton-Brown(2008)]{shoham2008multiagent}
Yoav Shoham and Kevin Leyton-Brown.
\newblock \emph{Multiagent systems: Algorithmic, game-theoretic, and logical
  foundations}.
\newblock Cambridge University Press, 2008.

\bibitem[Son et~al.(2019)Son, Kim, Kang, Hostallero, and Yi]{son2019qtran}
Kyunghwan Son, Daewoo Kim, Wan~Ju Kang, David~Earl Hostallero, and Yung Yi.
\newblock Qtran: Learning to factorize with transformation for cooperative
  multi-agent reinforcement learning.
\newblock In \emph{International Conference on Machine Learning}, pages
  5887--5896. PMLR, 2019.

\bibitem[Wang et~al.(2020)Wang, Ren, Liu, Yu, and Zhang]{wang2020qplex}
Jianhao Wang, Zhizhou Ren, Terry Liu, Yang Yu, and Chongjie Zhang.
\newblock Qplex: Duplex dueling multi-agent q-learning.
\newblock \emph{arXiv preprint arXiv:2008.01062}, 2020.

\bibitem[Yang et~al.(2020)Yang, Wen, Wang, Chen, Shao, Mguni, and
  Zhang]{yang2020multi}
Yaodong Yang, Ying Wen, Jun Wang, Liheng Chen, Kun Shao, David Mguni, and
  Weinan Zhang.
\newblock Multi-agent determinantal q-learning.
\newblock In \emph{International Conference on Machine Learning}, pages
  10757--10766. PMLR, 2020.

\bibitem[Yu et~al.(2021)Yu, Velu, Vinitsky, Wang, Bayen, and
  Wu]{yu2021surprising}
Chao Yu, Akash Velu, Eugene Vinitsky, Yu~Wang, Alexandre Bayen, and Yi~Wu.
\newblock The surprising effectiveness of mappo in cooperative, multi-agent
  games.
\newblock \emph{arXiv preprint arXiv:2103.01955}, 2021.

\bibitem[Zhang et~al.(2018)Zhang, Yang, Liu, Zhang, and Basar]{zhang2018fully}
Kaiqing Zhang, Zhuoran Yang, Han Liu, Tong Zhang, and Tamer Basar.
\newblock Fully decentralized multi-agent reinforcement learning with networked
  agents.
\newblock In \emph{International Conference on Machine Learning}, pages
  5872--5881. PMLR, 2018.

\bibitem[Zhou et~al.(2020)Zhou, Luo, Villella, Yang, Rusu, Miao, Zhang, Alban,
  Fadakar, Chen, et~al.]{zhou2020smarts}
Ming Zhou, Jun Luo, Julian Villella, Yaodong Yang, David Rusu, Jiayu Miao,
  Weinan Zhang, Montgomery Alban, Iman Fadakar, Zheng Chen, et~al.
\newblock Smarts: Scalable multi-agent reinforcement learning training school
  for autonomous driving.
\newblock \emph{arXiv preprint arXiv:2010.09776}, 2020.

\end{thebibliography}

\clearpage
\appendix
\onecolumn

{\Large{\textbf{Appendix}}}

\section{Reward Structure of Large-Scale Matrix Game}\label{section:son_matrix_game}
Reward matrix of the hard matrix game proposed by \citet{son2019qtran}. The game contains multiple stable points and a strongly attracting equilibrium: $(A, A)$ \citep{son2019qtran}.

\begin{table}[htb]
    \centering
    \label{table:son_matrix_game}
    \begin{tabular}{|c||c|c|c|}
        \hline
        \backslashbox{Agent $1$}{Agent $2$} & $A$ & $B$ & $C$\\
        \hline
        \hline
        $A$ & $8$ & $-12$ & $-12$ \\
        \hline
        $B$ & $-12$ & $0$ & $0$ \\
        \hline
        $C$ & $-12$ & $0$ & $0$ \\
        \hline
    \end{tabular}
\end{table}

An example reward structure of the game with more agents ($3$ agents) is as below. The letters in the left-most column of each panel (starting from the second row) represent the actions of Agent $1$, the letters on the top-row of each panel (starting from the second column) represent the actions of Agent $2$, and the top-left letter of each panel represents the action of Agent $3$. The reward structure for more agents is a straight-forward extrapolation of the below pattern.

\begin{tikzpicture}[every node/.style={anchor=north east,fill=white,minimum width=1.4cm,minimum height=7mm}]
\matrix (mA) [draw,matrix of math nodes]
{
C & A & B & C \\
A & -12 & 0 & 0 \\
B & 0 & 0 & 0 \\
C & 0 & 0 & 0 \\
};

\matrix (mB) [draw,matrix of math nodes] at ($(mA.south west)+(1.1,0.1)$)
{
B & A & B & C \\
A & -12 & 0 & 0 \\
B  & 0 & 0 & 0 \\
C  & 0 & 0 & 0 \\
};

\matrix (mC) [draw,matrix of math nodes] at ($(mB.south west)+(1.1,0.1)$)
{
A & A & B & C \\
A & 8 & -12 & -12 \\
B & -12 & 0 & 0 \\
C & -12 & 0 & 0 \\
};

\draw[dashed](mA.north east)--(mC.north east);
\draw[dashed](mA.north west)--(mC.north west);
\draw[dashed](mA.south east)--(mC.south east);
\end{tikzpicture}

An example of the reward structure of an larger action-space ($4$ actions) is as below. Even larger action-spaces are straight-forward extrapolations of the below pattern.
\begin{table}[htb]
    \centering
    \begin{tabular}{|c||c|c|c|c|}
        \hline
        \backslashbox{Agent $1$}{Agent $2$} & $A$ & $B$ & $C$ & $D$\\
        \hline
        \hline
        $A$ & $8$ & $-12$ & $-12$ & $-12$ \\
        \hline
        $B$ & $-12$ & $0$ & $0$ & $0$ \\
        \hline
        $C$ & $-12$ & $0$ & $0$ & $0$ \\
        \hline
        $D$ & $-12$ & $0$ & $0$ & $0$ \\
        \hline
    \end{tabular}
\end{table}

\clearpage

\section{Ablation Study - Number of Samples Used in Marginalization}
We ran an ablation study over the number of samples required to marginalise $Q_i$ to obtain $\hat{Q}$. We selected three tasks from Level-based Foraging at random: Foraging-5x5-2p-1f-coop, Foraging-10x10-5p-3f, and Foraging-10x10-3p-5f. We ran PERLA\_MAPPO 
on these maps with a range of values for $K$ for three seeds each. Figure \ref{fig:K_ablation} shows the area-under-the-curve averaged across seeds and across the three maps versus $K$. As expected, the general trend is that more samples yields better performance. PERLA\_MAPPO with $K=100$ is significantly better than PERLA\_MAPPO with $K=5$. 
However this trend discontinues beyond $K=250$ --- the algorithm performs slightly worse than $K=100$ indicating that increasing the number of samples does not improve performance.

\begin{figure}[h!]
    \centering
    \includegraphics[width=0.5\columnwidth]{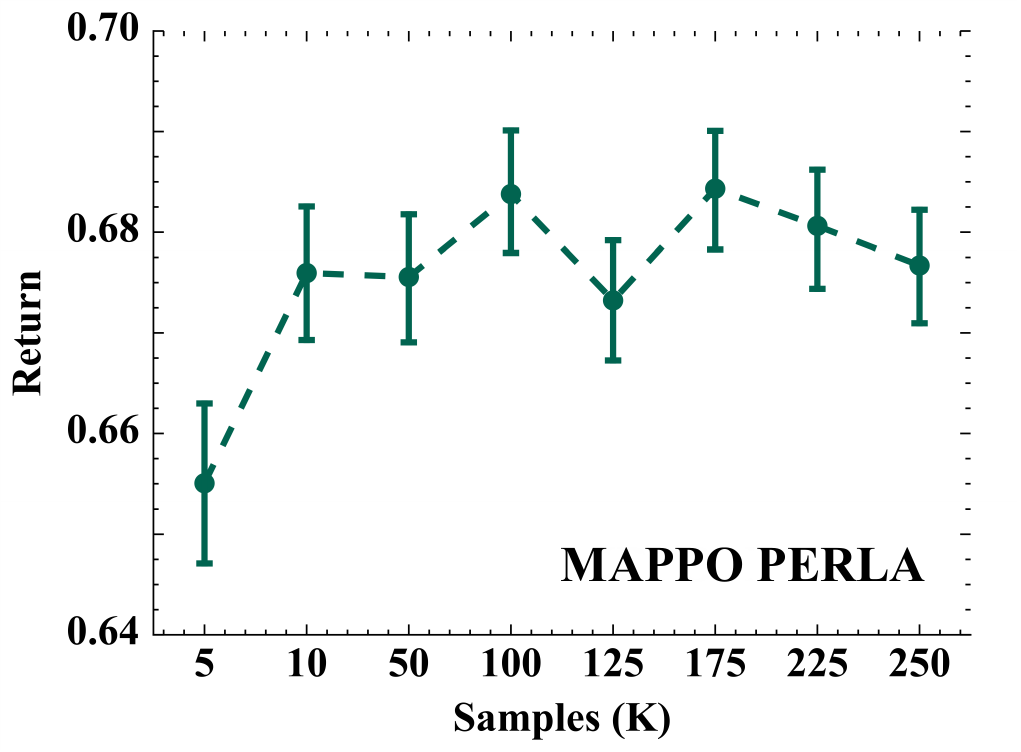}
    \caption{Ablation on $K$, the number of samples used to marginalise $Q_i$. As expected, the overall trend in the plots suggests that more samples are better.}
    \label{fig:K_ablation}
\end{figure}

\clearpage

\section{Learning Curves on Level-Based Foraging}\label{appendix:all_plots}
\begin{figure}[h!]
    \begin{center}
        \includegraphics[width=\textwidth]{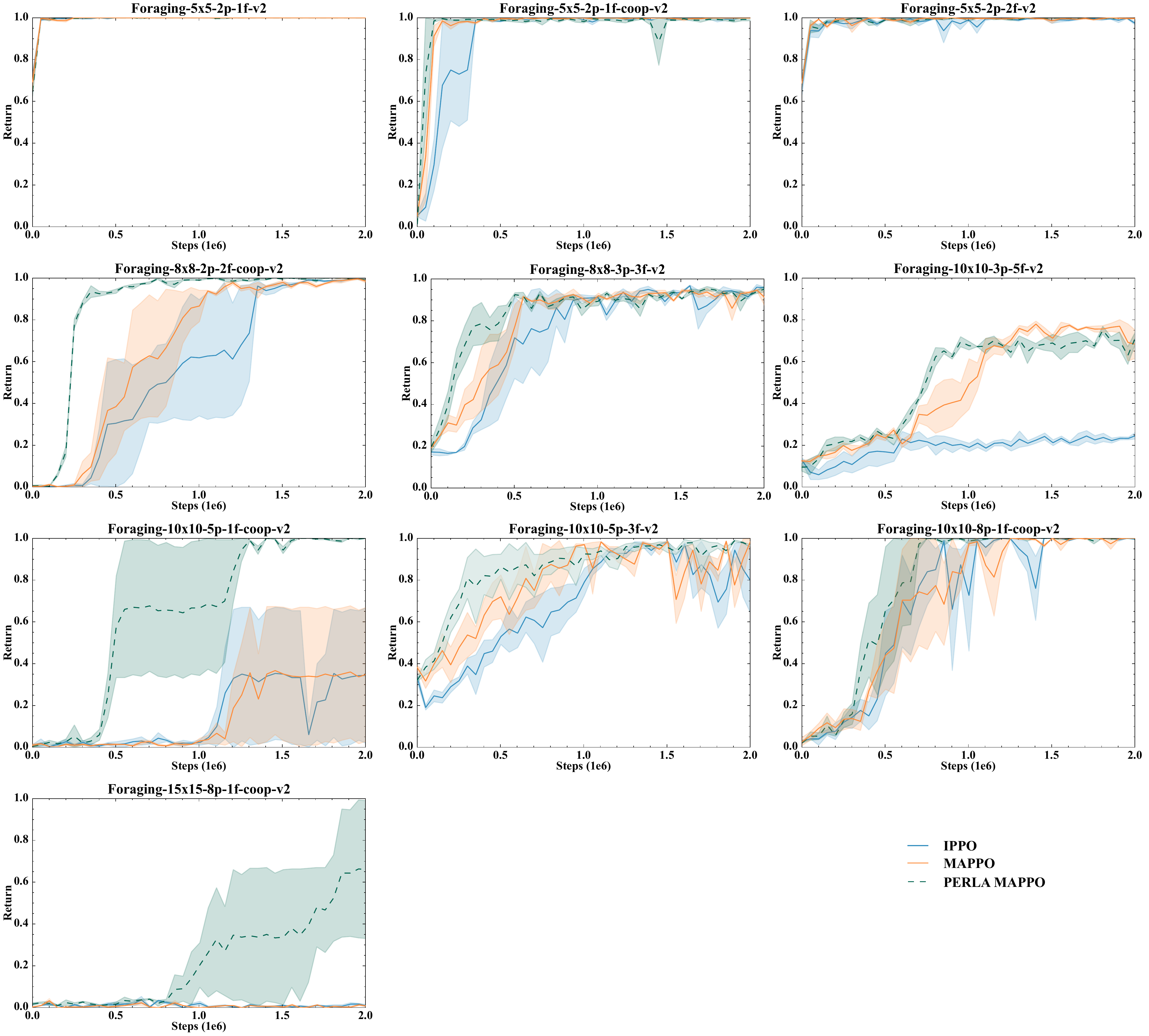}
        \caption{Learning curves of PERLA MAPPO and baselines on all tested \emph{Level-Based Foraging} maps. In all maps, PERLA MAPPO is outperform or is equivalent to the baselines. Moreover, in \emph{Foraging-15x15-8p-1f-v2} --- the map requiring the highest level of coordination (8 agents must cooperate to achieve the goal) --- vanilla MAPPO fails to learn entirely, while PERLA MAPPO successfully learns and approaches the optimal solution. This demonstrates that PERLA enables MAPPO to scale to scenarios in which many agents must cooperate to succeed.}
        \label{fig:LBF_IPPO}
    \end{center}
\end{figure}


\clearpage

\section{Learning Curves on Starcraft Multi-agent Challenge}\label{section:smac_all_maps}
\begin{figure}[h!]
    \begin{center}
        \includegraphics[width=0.95\textwidth]{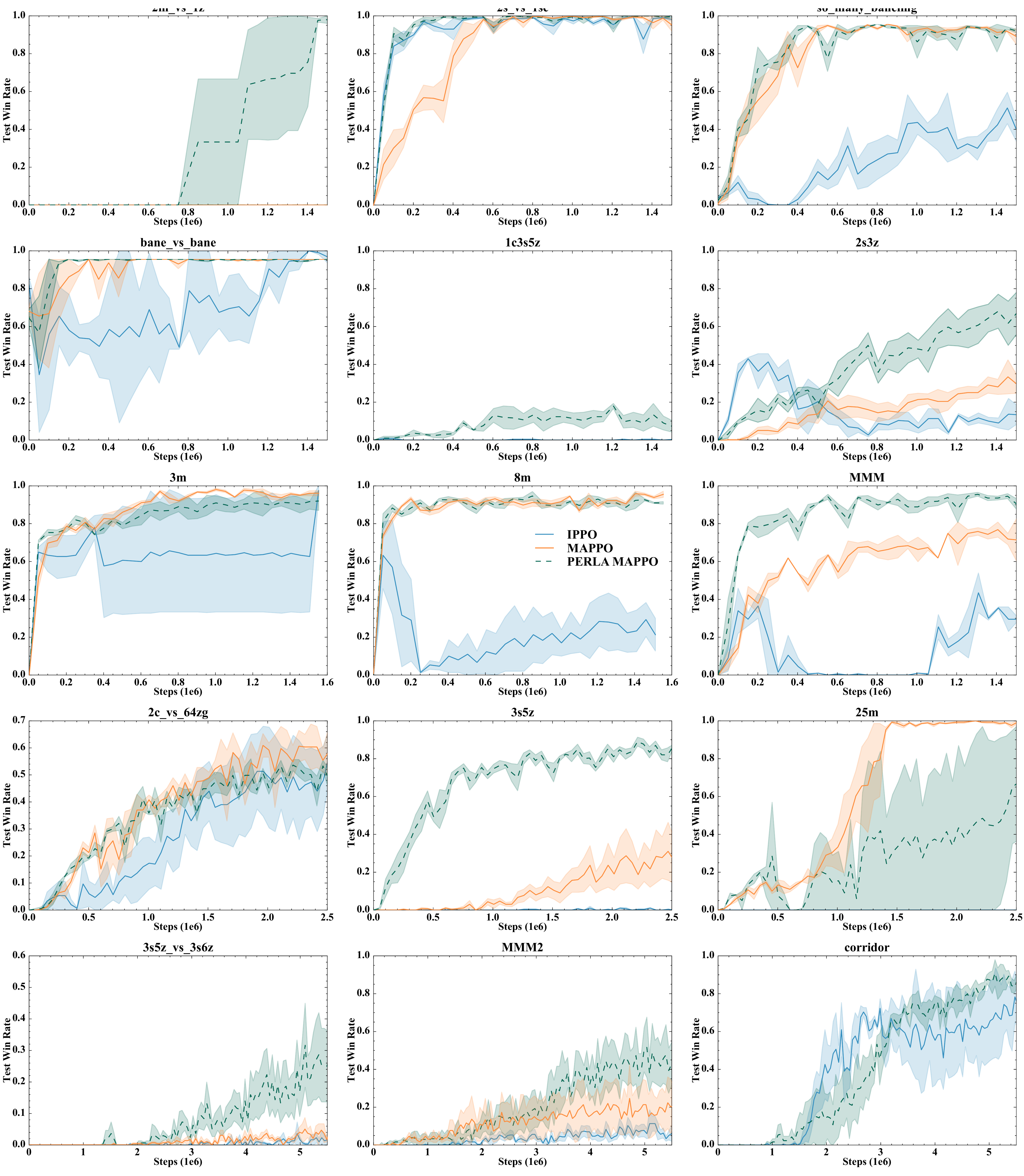}
        \caption{Learning curves of PERLA MAPPO and baselines on all tested \emph{StarCraft Multi-agent Challenge} maps. In all except one map  PERLA MAPPO outperforms or matches all baselines. Notably, PERLA MAPPO even out-performs IPPO on \emph{corridor}, a map where it is well known that MAPPO struggles and IPPO is often more effective. This implies that PERLA may even help algorithms over-come failure modes.}
        \label{figure:smac_all_maps_learning_curves}
    \end{center}
\end{figure}

\clearpage


\section{Computing Resources and Hyperparmeters}
\subsection{Computing Resources}
All experiments (including hyperparameter search) were run on machines equipped as desribed in Table \ref{tab:compute}. We used multiple such machines in parallel to complete all runs.

\begin{table}[h!]
    \centering
    \begin{tabular}{l|l}
        \textbf{Component} & \textbf{Description} \\
        \hline
        CPU & Intel Core i9-9900X CPU @ 3.50GHz \\
        GPU & Nvidia RTX 2080  \\
        Memory & 64 GB DDR4 \\
    \end{tabular}
    \caption{Specifications of machines used to run empirical component of this paper.}
    \label{tab:compute}
\end{table}

\subsection{Hyperparameters}
We used the implementation of \citep{papoudakis2021benchmarking} for the baseline algorithms as well as the foundation upon which we implemented the \emph{PERLAISED} version of IPPO and MAPPO. For the baseline algorithms, we re-used the optimal hyperparameters reported by \citep{papoudakis2021benchmarking}. To optimise PERLA\_IPPO and PERLA\_MAPPO, for each environment, we sampled a small number of tasks and used simple grid-search to optimise selected hyperparameters. The combination of hyperparameters that achieved maximum Return or Test Win-Rate were then used in the runs of all remaining tasks. Table \ref{tab:hyperparams} shows the hyperparameters optimised and the range of values evaluated.

\begin{table}[h]
    \centering
    \begin{tabular}{l|l}
        \textbf{Hyperparameter} & \textbf{Range} \\
        \hline
        Hidden Dimension & [64, 128, 256] \\
        Learning Rate & [3e-4, 5e-4] \\
        Network Type & [Fully-connected, Gated Recurrent Network] \\
        Reward Standardisation & [False, True] \\ 
        Entropy Coefficient & [1e-3 1e-2 1e-1] \\
        Target network update & [Hard update every 200 steps, Soft update of 0.01 per step] \\
        Number of steps for calculating Return & [5, 10] \\
        Number of PERLA samples (K) & [10, 100] \\
    \end{tabular}
    \caption{Hyperparameters and range of values evaluated to optimise \emph{PERLAISED} variants of IPPO and MAPPO.}
    \label{tab:hyperparams}
\end{table}

\section{LBF experiments} \label{sec:lbf_experiments_appendix}
\subsection{Maps}
We used the following maps in Level-based Foraging: \emph{Foraging-5x5-2p-1f-v2,  Foraging-5x5-2p-1f-coop-v2, Foraging-5x5-2p-2f-v2, Foraging-8x8-2p-2f-coop-v2, Foraging-8x8-3p-3f-v2, Foraging-10x10-5p-1f-coop-v2, Foraging-10x10-3p-5f-v2, Foraging-10x10-5p-3f-v2, Foraging-10x10-8p-1f-coop-v2, Foraging-15x15-8p-1f-coop-v2}

\subsection{Selected Hyperparameters}
We used the hyperparameters given in Table \ref{tab:hyperparams-lbf} for all runs of LBF.
\begin{table}[h]
    \centering
    \begin{tabular}{l|l}
        \textbf{Hyperparameter} & \textbf{Value} \\
        \hline
        Hidden Dimension & [256] \\
        Learning Rate & [5e-4] \\
        Network Type & [Fully-connected] \\
        Reward Standardisation & [False] \\ 
        Entropy Coefficient & [1e-3] \\
        Target network update & [Soft update of 0.01 per step] \\
        Number of steps for calculating Return & [10] \\
        Number of PERLA samples (K) & [100] \\
    \end{tabular}
    \caption{Optimised hyperparameters of PERLA MAPPO on LBF.}
    \label{tab:hyperparams-lbf}
\end{table}

\clearpage

\section{SMAC Experiments}\label{section:smac_details_appendix}
\subsection{Maps}
We used the following maps in StarCraft Multiagent Challenge.\\
\textbf{Easy} - \emph{2m\_vs\_1z, 2s\_vs\_1sc, so\_many\_baneling, bane\_vs\_bane, 1c3s5z, 2s3z, 3m, 8m, MMM}\\
\textbf{Hard} - \emph{2c\_vs\_64zg, 3s5z, 25m}\\
\textbf{Very Hard} - \emph{3s5z\_vs\_3s6z, MMM2, corridor}

\subsection{Selected Hyperparameters}
Hyperparameters used for Multi-agent PPO in the SMAC domain.
\begin{table}[h!]
    \begin{adjustbox}{width=0.4\columnwidth, center}
        \begin{tabular}{c | c } 
            \toprule
            Hyperparameters & value\\
            \midrule
            actor lr & 1e-3\\ 
            critic lr & 5e-4\\
            gamma & 0.99\\
            batch size & 3200\\
            num mini batch & 1\\
            PPO epoch & 10\\
            PPO clip param & 0.2\\
            entropy coef & 0.01\\
            optimiser & ADAM\\
            opti eps & 1e-5\\
            max grad norm & 10\\
            actor network & mlp\\
            hidden layper & 1\\
            hidden layer dim & 64\\
            activation & ReLU\\
            gain & 0.01\\
            eval episodes & 32\\
            use huber loss & True\\
            rollout threads & 32\\
            episode length & 100\\
            Number of PERLA samples (K) & 100 \\

            \bottomrule
        \end{tabular}
    \end{adjustbox}
\end{table}

\clearpage

\section{Matrix Game Experiments}
\subsection{Selected Hyperparameters}
Hyperparameters used for Multi-agent PPO in the matrix game domain.

\begin{table}[h!]
    \begin{adjustbox}{width=0.3\columnwidth, center}
        \begin{tabular}{c | c } 
            \toprule
            Hyperparameters & value\\
            \midrule
            actor lr & 1e-4\\ 
            gamma & 0.99\\
            batch size & 64\\
            num mini batch & 1\\
            PPO epoch & 5\\
            PPO clip param & 0.2\\
            entropy coef & 0.01\\
            optimiser & ADAM\\
            opti eps & 1e-5\\
            max grad norm & 10\\
            actor network & mlp\\
            hidden layper & 1\\
            hidden layer dim & 64\\
            activation & ReLU\\
            gain & 0.01\\
            eval episodes & 1\\
            use huber loss & False\\
            rollout threads & 4\\
            episode length & 1\\
            \bottomrule
        \end{tabular}
    \end{adjustbox}
\end{table}

\clearpage

\section{Multi-agent Mujoco Experiments}
Hyperparameters used for Multi-agent PPO in the Multi-Agent MuJoCo domain.
\begin{table}[h!]
    \begin{adjustbox}{width=0.6\columnwidth, center}
        \begin{tabular}{c | c c c } 
            \toprule
            Hyperparameters & Hopper(3x1) & Swimmer(2x1) & Walker(2x3)\\ 
            \midrule
            actor lr & 5e-4 & 5e-4 & 5e-4\\ 
            critic lr & 5e-3 & 5e-3 &  5e-3\\
            lr decay & 1 & 1 & 1\\
            $\gamma$ & 0.99 & 0.99 & 0.99\\
            batch size & 4000 & 4000  & 4000\\
            num mini batch & 40 & 40 & 40\\
            PPO epoch & 5 & 5 & 5\\
            PPO clip param & 0.2 & 0.2  & 0.2\\
            entropy coef & 0.001 & 0.001  & 0.001\\
            optimiser & RMSProp & RMSProp  & RMSProp\\
            momentum & 0.9 & 0.9  & 0.9\\
            optim eps & 1e-5 & 1e-5 & 1e-5\\
            max grad norm & 0.5 & 0.5 & 0.5\\
            actor network & mlp & mlp  & mlp\\
            hidden layer & 2 & 2 & 2\\
            hidden layer dim & 128 & 128 & 128\\
            activation & ReLU & ReLU & ReLU\\
            eval episodes & 10 & 10  & 10\\
            rollout threads & 4 & 4 & 4\\
            episode length & 1000 & 1000 & 1000\\
            \bottomrule
        \end{tabular}
    \end{adjustbox}
\end{table}

\clearpage
\section{Theoretical Discussion on Value-based critic} \label{ap:vcritic}
In practical implementations of Actor-Critic algorithms, it is popular to train a value-based critic rather than a state-action critic, as often the underlying network tends to be easier to train. In this section we show the consistency of such approach with our framework and theoretical analysis. First we note that by definition of state-action value function we have:
\begin{equation*}
 \mathbb{E}[\sum_{l=t}^{\infty} \gamma^{l-t} r^l |s^t=s,a_i^t = a_i,\va_{-i}^t = a_{-i} ] = 
\end{equation*}
\begin{equation*}
     \mathbb{E}[\sum_{l=t}^{t+m} \gamma^{l-t} r^l |s^t=s,a_i^t = a_i,\va_{-i}^t = a_{-i} ] + \mathbb{E}[\mathbb{E}[\sum_{l=t+m}^{\infty} \gamma^{l-t} r_l|s^l, \va_{-i}^l] |s^t=s,a_i^t = a_i,\va_{-i}^t = a_{-i} ] =
\end{equation*}
\begin{equation*}
    = \mathbb{E}[\sum_{l=t}^{t+m} \gamma^{l-t} r^l |s^t=s,a_i^t = a_i,\va_{-i}^t = \va_{-i} ] + \mathbb{E}[V(s^l, \va_{-i}^l) |s^t=s,a_i^t = a_i,\va_{-i}^t = a_{-i}] ,
\end{equation*}
where we introduced the value critic $V(s^l, \va_{-i}^l) = \mathbb{E}[\sum_{l=t+m}^{\infty} \gamma^{l-t} r_l|s^l, \va_{-i}^l]$.
We can thus approximate the expectations by one-off trajectory sample from the environment giving raise to the $m$-step value function critic:
\begin{equation*}
     Q(s^t,a_i^t,\va_{-i}^t) \approx \sum_{l=t}^{t+m} \gamma^{l-t} r^l + \gamma^{m} V(s^m, \va_{-i}^m) ,
\end{equation*}
where the trajectory is sampled following current agents' policies, after having selected joint action $(a_i^t,\va_{-i}^t)$ at state $s^t$. Observe that while the rewards $r^l$ and state after $m$ steps $s^m$ are sampled from the environment, the joint action of other agents after $m$ steps $\va_{-i}^m$ is sampled from their joint policy. Hence, we can obtain more samples  of $\va_{-i}^m$, without performing additional steps in the environment. Let $\va_{-i}^{m(j)}$ be the $j$th sample from $\vpi_{-i}(s^m)$, then we can simply marginalise the value-function critic as follows:
\begin{equation*}
    \tilde{V}(s^t) = \frac{1}{k} \sum_{j=1}^k V(s^t, \va_{-i}^{l(j)}).
\end{equation*}
By tower property of expectations and the properties of Monte-Carlo estimators we have that:
\begin{equation*}
   \mathbb{E}[\hat{V}(s^m) |s^t=s,a_i^t = a_i,\va_{-i}^t = a_{-i}] =  \mathbb{E}[\mathbb{E}[\hat{V}(s^m)|s^m]|s^t=s,a_i^t = a_i,\va_{-i}^t = a_{-i}] = 
\end{equation*}
\begin{equation*}
    = \mathbb{E}[\mathbb{E}[V(s^m)|s^m]|s^t=s,a_i^t = a_i,\va_{-i}^t = a_{-i}] = \mathbb{E}[V(s^m)|s^t=s,a_i^t = a_i,\va_{-i}^t = a_{-i}]
\end{equation*}

Hence both the estimators $ Q(s^t,a_i^t,\va_{-i}^t) = \sum_{l=t}^{t+m} \gamma^{l-t} r^l + V(s^m, \va_{-i}^m)$ and $ \tilde{Q}(s^t,a_i^t,\va_{-i}^t) = \sum_{l=t}^{t+m} \gamma^{l-t} r^l + \tilde{V}(s^m, \va_{-i}^m)$ are equal in expectation, and since $Q(s^t,a_i^t,\va_{-i}^t)$ is unbiased, so is $\tilde{Q}(s^t,a_i^t,\va_{-i}^t)$. Observe that for $C = A + B$, we have:
\begin{equation*}
    \textrm{Var}(C) =  \textrm{Var}(A) + \textrm{Var}(B) + 2\textrm{Cov}(A, B) \le \textrm{Var}(A) + \textrm{Var}(B) + 2\sqrt{\textrm{Var}(A)\textrm{Var}(B)}
\end{equation*}
We also note that $\textrm{Var}(\tilde{V}(s, \va_{-i})) = \frac{\textrm{Var}(V(s, \va_{-i}))}{k}$. Thus we obtain the following upper bounds on the variance of $Q$ and $\tilde{Q}$:
\begin{equation*}
    \textrm{Var}(Q(s^t,a_i^t,\va_{-i}^t)) \le \textrm{Var} \left (\sum_{l=t}^{t+m} \gamma^{l-t} r^l \right ) + \textrm{Var}(V(s^l, \va_{-i}^l)) + 2\sqrt{\textrm{Var}(V(s^l, \va_{-i}^l)) \textrm{Var} \left (\sum_{l=t}^{t+m} \gamma^{l-t} r^l \right)}
\end{equation*}
\begin{equation*}
    \textrm{Var}(\tilde{Q}(s^t,a_i^t,\va_{-i}^t)) \le \textrm{Var} \left (\sum_{l=t}^{t+m} \gamma^{l-t} r^l \right ) + \frac{\textrm{Var}(V(s^l, \va_{-i}^l))}{k} + 2\sqrt{\frac{\textrm{Var}(V(s^l, \va_{-i})) \textrm{Var} \left (\sum_{l=t}^{t+m} \gamma^{l-t} r^l \right)}{k}}
\end{equation*}
Hence we conclude that even in the case of a $m$-step critic, the marginalisation property still produces a smaller bound on the variance.
\clearpage
\section{Proofs of Theoretical Results}

\subsection*{Proof of Theorem \ref{th:mcfvar}} \label{ap:mcfvar}
Let us first present a Lemma, which we use to prove the Theorem \ref{th:mcfvar}.

\begin{lemma} \label{lemma:mcfvar}
Given $N$ random variables $( x_i)_{i=1}^N$, where $x_i: \Omega \to \mathcal{X}_i$ and a measurable function $f: \times_{i=1}^N \mathcal{X}_i \to \mathbb{R}$, if we define $\tilde{f}(x_1,\dots,x_M) := \mathbb{E}[f(x_1,\dots,x_N)|x_1, \dots, x_M]$ for some $M \le N$, we have that:
\begin{equation} \label{eq:variance_improvement}
     \textrm{Var}(f(x_1,\dots,x_N)) \ge \textrm{Var}(\tilde{f}(x_1, \dots, x_M)).
\end{equation}
Moreover, for a $k$-sample Monte-Carlo estimator of $\tilde{f}$ given by: 
$$\hat{f}(x_1,\dots,x_M) = \frac{1}{k} \sum_{i=1}^k f(x_1, \dots, x_M , x_{M+1}^{(i)},\dots,x_N^{(i)}),$$
where $x_j^{(i)}$ is the $i^{th}$ sample of $x_j$ we have:
\begin{equation} \label{eq:mc_variance}
    \textrm{Var}(\hat{f}(x_1,\dots,x_M)) = \frac{1}{k} \textrm{Var}(f(x_1, \dots, x_N)) + \frac{k-1}{k} \textrm{Var}(\tilde{f}(x_1,\dots,x_M)).
\end{equation}
\end{lemma}

\begin{proof}
Using the law of total variance we can decompose the variance of $\hat{f}$ as follows:
\begin{equation*}
    \textrm{Var}(f(x_1,\dots,x_N)) = \mathbb{E}[\textrm{Var}(f(x_1,\dots,x_N)| x_1 ,\dots, x_M)] + \textrm{Var}(\mathbb{E}[f(x_1,\dots,x_N)| x_1 ,\dots, x_M]) 
\end{equation*}
We can now observe that the conditional expectation inside the variance in second term is just $\tilde{f}$. We therefore get:
\begin{equation} \label{eq:var_decomp_f}
    \textrm{Var}(f(x_1,\dots,x_N)) = \mathbb{E}[\textrm{Var}(f(x_1,\dots,x_N)| x_1 ,\dots, x_M)] + \textrm{Var}(\tilde{f}(x_1,\dots,x_M))
\end{equation}
The first term is an expectation of variance and therefore cannot be negative, this proves the statement in Equation \ref{eq:variance_improvement}.  Applying the law of total variance, again, this time for $\hat{f}(x_1,\dots,x_M)$ we get:
\begin{equation*}
    \textrm{Var}(\hat{f}(x_1,\dots,x_M)) = \mathbb{E}[\textrm{Var}(\hat{f}(x_1,\dots,x_M)|x_1,\dots,x_M)] + \textrm{Var}(\mathbb{E}[\hat{f}(x_1,\dots,x_M)|x_1,\dots,x_M])
\end{equation*}
Note that $\hat{f}(x_1,\dots,x_M)$ is a random variable dependent on $(x_{M+1},\dots,x_N)$. Since we estimate $\hat{f}$ by Monte-Carlo method with $k$ samples, we get that the expectation must be equal to the true value $\tilde{f}$ and variance is equal to the variance of samples divided by $k$. We therefore get:
\begin{equation*} 
    \textrm{Var}(\hat{f}(x_1,\dots,x_M)) = \frac{1}{k}\mathbb{E}[\textrm{Var}(f(x_1,\dots,x_N)|x_1,\dots,x_M)] + \textrm{Var}(\tilde{f}(x_1,\dots,x_M))
\end{equation*}
Substituting an expression for $\mathbb{E}[\textrm{Var}(f(x_1,\dots,x_N)|x_1,\dots,x_M)]$ from Equation \ref{eq:var_decomp_f}, we get:
\begin{equation*}
    \textrm{Var}(\hat{f}(x_1,\dots,x_M)) = \frac{1}{k}\textrm{Var}(f(x_1,\dots,x_N)) + \frac{k-1}{k}\textrm{Var}(\tilde{f}(x_1,\dots,x_N))
\end{equation*}
which proves the statement in Equation \ref{eq:mc_variance}.
\end{proof}

\mcfvar*
\begin{proof}

Taking the Q-function as $f$, the action $a_i$ of the $i$th  agent and the state $s$ as random variables $x_1,\dots,x_M$, and actions of other agents $\va_{-i}$ as remaining variables $x_{M+1}, \dots, x_N$ we immediately obtain the statement of the Theorem from Lemma \ref{lemma:mcfvar}.
\end{proof}
\section*{Proof of Theorem \ref{th:equal_exp}} \label{ap:equal_exp}

\equalexp*
\begin{proof}

Using the tower property of expectation we have:
\begin{align*}
    \mathbb{E}\left[\vg_{i}^{\textrm{C}}\right] &= \mathbb{E}\left[\mathbb{E}\left[\sum_{t=0}^{\infty}\gamma^{t}  Q_i\left(s^{t}, a_{i}^{t} \right)\nabla_{\vtheta_{i}}\log \pi_{i}\left(a_{i}^{t}\big|s^{t}\right)|a_i^t,s^t\right]\right]
    \\&= \mathbb{E}\left[ \sum_{t=0}^{\infty}\gamma^{t}\mathbb{E}[ Q_i\left(s^{t}, a_{i}^{t} \right)|a_i^t, s^t]\nabla_{\vtheta_{i}}\log \pi_{i}\left(a_{i}^{t}\big|s^{t}\right)\right] 
    \\&=\mathbb{E}\left[ \sum_{t=0}^{\infty}\gamma^{t} \mathbb{E}[\hat{Q}_i\left(s^{t}, a_{i}^{t} \right)|a_i^t,s^t]\nabla_{\vtheta_{i}}\log \pi_{i}\left(a_{i}^{t}\big|s^{t}\right)\right] 
    \\&
    = \mathbb{E}\left[ \sum_{t=0}^{\infty}\gamma^{t} \hat{Q}_i\left(s^{t}, a_{i}^{t} \right)\nabla_{\vtheta_{i}}\log \pi_{i}\left(a_{i}^{t}\big|s^{t}\right)\right] = \mathbb{E}\left[\vg_{i}^{\textrm{P}}\right],
\end{align*}
where the third equality is due to Monte-Carlo estimates being unbiased.
\end{proof}

\section*{Proof of Theorem \ref{th:sharingtdbound}} \label{ap:sharingdtbound}
\sharingdtbound*
\begin{proof}
Using the law of total variance on the variance of $j$-th component of PERLA estimator we get:
\begin{equation*}
    \textrm{Var}(g^S_{i,j}) = \mathbb{E}[\textrm{Var}(g^S_{i,j}|a_i^t,s^t)] + \textrm{Var}(\mathbb{E}[g^S_{i,j}|a_i^t,s^t])  
\end{equation*}
Analogously for the $j$-th component of the decentralised estimator we get:
\begin{equation*}
    \textrm{Var}(g^D_{i,j}) = \mathbb{E}[\textrm{Var}(g^D_{i,j}|a_i^t,s^t)] + \textrm{Var}(\mathbb{E}[g^D_{i,j}|a_i^t,s^t])  
\end{equation*}
Using the fact that Monte-Carlo estimates are unbiased and that critics are perfect, we get:
\begin{equation*}
   \textrm{Var}(\mathbb{E}[g^S_{i,j}|a_i^t,s^t])   =  \textrm{Var}\left(\sum_{t=0}^{\infty}\gamma^{t}\nabla_{\theta_{i,j}}\log \pi_{i}\left(a^{i}_{t}\big|s_{t}\right)\mathbb{E}[\hat{Q}\left(s_{t}, a_{i}^{t}\right)|a_i^t,s^t]\right) = \textrm{Var}(\mathbb{E}[g^D_{i,j}|a_i^t,s^t])  
\end{equation*}
We therefore obtain:
\begin{align*}
    &\textrm{Var}(g^S_{i,j}) - \textrm{Var}(g^D_{i,j})  \\&=
    \mathbb{E}[\textrm{Var}(g^S_{i,j}|a_i^t,s^t)] - \mathbb{E}\left[\textrm{Var}(g^D_{i,j}|a_i^t,s^t)\right] 
    \\&=\mathbb{E}\left[  \sum_{t=0}^{\infty}\gamma^{2t}(\nabla_{\theta_{i,j}}\log \pi_{i}\left(a^{i}_{t}\big|s^{t}\right))^2(\textrm{Var}(\hat{Q}\left(s^{t}, a_{i}^{t} \right) |a_i^t,s^t)-\textrm{Var}(\tilde{Q}\left(s^{t}, a_{i}^{t} \right) |a_i^t,s^t))\right]. 
\end{align*}
Since $\tilde{Q}_i(s^t,a_i^t)$ is a deterministic function given $a_i^t$ and $s_t$, the expression above simplifies to:
\begin{equation*}
   \textrm{Var}(g^S_{i,j}) - \textrm{Var}(g^D_{i,j}) = \mathbb{E}\left[  \sum_{t=0}^{\infty}\gamma^{2t}(\nabla_{\theta_{i,j}}\log \pi_{i}\left(a^{i}_{t}\big|s_{t}\right))^2 \textrm{Var}(\hat{Q}\left(s^{t}, a^{t}_i \right)|a_i^t,s^t)\right] 
\end{equation*}
Since $\hat{Q}_i$ is a Monte-Carlo approximation of $\tilde{Q}_i$ with $k$ samples of $Q_i$ we get:
\begin{equation*}
   \textrm{Var}(g^S_{i,j}) - \textrm{Var}(g^D_{i,j})  = \mathbb{E}\left[  \sum_{t=0}^{\infty}\gamma^{2t}(\nabla_{\theta_{i,j}}\log \pi_{i}\left(a^{i}_{t}\big|s_{t}\right))^2 \frac{1}{k} \textrm{Var}(Q\left(s^{t}, a^{t} \right)|a_i^t,s^t)\right] 
\end{equation*}
We can now sum over all components of the gradient vector to obtain:
\begin{align}
    &\textrm{Var}(\vg_{i}^{\textrm{C}}) - \textrm{Var}(\vg_{i}^{\textrm{D}}) = \sum_{j=1}^d \mathbb{E}\left[  \sum_{t=0}^{\infty}\gamma^{2t}(\nabla_{\theta_{i,j}}\log \pi_{i}\left(a^{i}_{t}\big|s_{t}\right))^2 \frac{1}{k} \textrm{Var}(Q\left(s^{t}, a^{t} \right)|a_i^t,s^t)\right]
    \\&= \mathbb{E}\left[  \sum_{t=0}^{\infty}\gamma^{2t}\lVert \nabla_{\theta_{i,j}}\log \pi_{i}\left(a^{i}_{t}\big|s_{t}\right)\rVert^2 \frac{1}{k} \textrm{Var}(Q\left(s^{t}, a^{t} \right)|a_i^t,s^t)\right] \label{eq:uncomplete_bound_var}
   \\&  \le \frac{B_i^2}{k} \frac{1}{1-\gamma^2} \mathbb{E}\left[\textrm{Var}(Q\left(s^{t}, a^{t} \right)|a_i^t,s^t)\right].
\end{align}
We can now upper-bound the variance of the Q-function as follows:
\begin{equation*}
    \textrm{Var}(Q\left(s^{t}, a^{t} \right)|a_i^t,s^t) = \mathbb{E}\left[Q\left(s^{t}, a^{t} \right)^2|a_i^t,s^t] - \mathbb{E}[Q\left(s^{t}, a^{t} \right)|a_i^t,s^t\right]^2 
\end{equation*}
\begin{equation} \label{eq:q_func_var_ub}
     \le \mathbb{E}\left[Q\left(s^{t}, a^{t} \right)^2|a_i^t,s^t\right] \le C_i^2
\end{equation}
Combining Equations \ref{eq:uncomplete_bound_var} and \ref{eq:q_func_var_ub} completes the proof.

\end{proof}

\section*{Proof of Theorem \ref{th:ctdedtbound}}
\ctdedtbound*
\begin{proof}
Similar to the proof of Theorem \ref{th:sharingtdbound}, because we assume perfect critics, we have:
\begin{equation*}
    \textrm{Var}(\mathbb{E}[g^C_{i,j}|a_i^t,s^t])   =  \textrm{Var}\left(\sum_{t=0}^{\infty}\gamma^{t}\nabla_{\theta_{i,j}}\log \pi_{i}\left(a^{i}_{t}\big|s_{t}\right)\mathbb{E}[Q\left(s_{t}, a_{i}^{t}\right)|a_i^t,s^t]\right) = \textrm{Var}(\mathbb{E}[g^D_{i,j}|a_i^t,s^t])  
\end{equation*}
We can therefore follow the proof of Theorem \ref{th:sharingtdbound}, using the law of total variance to obtain that:
\begin{equation*}
    \textrm{Var}(\vg_{i}^{\textrm{C}}) - \textrm{Var}(\vg_{i}^{\textrm{D}}) \le \mathbb{E}[  \sum_{t=0}^{\infty}\gamma^{2t}\lVert \nabla_{\theta_{i,j}}\log \pi_{i}\left(a^{i}_{t}\big|s_{t}\right)\rVert^2 \textrm{Var}(Q\left(s^{t}, a^{t} \right)|a_i^t,s^t)]
\end{equation*}
\begin{equation*}
     \le B_i^2 \frac{1}{1-\gamma^2} \mathbb{E}[\textrm{Var}(Q\left(s^{t}, a^{t} \right)|a_i^t,s^t)]
\end{equation*}
We can now bound the variance of Q-function using the inequality in Equation \ref{eq:q_func_var_ub}, which completes the proof.
\end{proof}

\section*{Proof of Theorem \ref{th:convergence}} \label{ap:convergence}

\convergence*

\begin{proof}
We prove the convergence by showing that a multi-agent Actor-Critic algorithm based on $\vg_i^P$ or $\vg_i^C$ can be expressed as a special case of a single agent Actor-Critic algorithm. The convergence of single agent Actor-Critic algorithm has been established by \citep{konda1999actor}. 
First, let us define a joint gradient as a vector consisting of concatenated gradient vectors for each agent. We call them $\vg^C = ((\vg^C_1)^T,\dots,(\vg^C_N)^T)^T$ and $\vg^S = ((\vg^S_1)^T,\dots,(\vg^S_N)^T)^T$ for the centralised and PERLA estimator respectively (we denote the number of agents by $N$). Let us also define the joint policy for all agents as $\vpi(\va|s) = \prod_{i=1}^N \pi_i(a_i|s)$. Since only the policy of the $i$th agent depends on set of parameters $\vtheta_i$, we have that:
\begin{equation*}
    \nabla_{\vtheta_i} \log \vpi(\va|s) = \nabla_{\vtheta_i} \log \prod_{i=1}^N \pi_i(a_i|s)  = \nabla_{\vtheta_i} \sum_{i=1}^N \log \pi_i(a_i|s) = \nabla_{\vtheta_i} \log \pi_i(a_i|s)
\end{equation*}
We can therefore write:
\begin{equation*}
     \nabla_{\vtheta} \log \vpi(\va|s) = ((\nabla_{\vtheta_1} \log \pi_1(a_1|s))^T,\dots, (\nabla_{\vtheta_N} \log \pi_N(a_N|s))^T)^T
\end{equation*}
This enables us to express the centralised estimator as:
\begin{equation*}
    \vg^C = \sum_{t=0}^{\infty}\gamma^{t}Q(s^{t},\va^t)\nabla_{\vtheta}\log \vpi(\va^{t}\big|s^{t})
\end{equation*}
Let us now define a single agent using policy $\vpi(\va^t|s^t)$ and taking a multidimensional action $\va^t$ at each step. In such a case, $\vg^C$ is an unbiased policy gradient estimate for that agent, which proves the convergence for an Actor-Critic algorithm using $\vg^C$. We have already established in Theorem \ref{th:equal_exp} that $\mathbb{E}[\vg^C_i] = \mathbb{E}[\vg^S_i]$, which means $\mathbb{E}[\vg^C] = \mathbb{E}[\vg^S]$. Therefore $\vg^S$ is also an unbiased estimate of this special case of single-agent Actor-Critic, which proves its convergence by the same argument.
\end{proof}

\clearpage

\section{Ablation on number of agents}
\begin{figure}[h]
    \centering
    \includegraphics{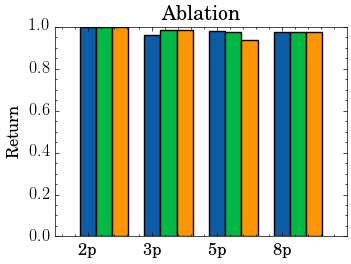}
    \caption{Final mean return against the number of players increases in LBF environment. Blue indicates 5 samples, green indicates 25 and orange corresponds to 125 samples.}
    \label{fig:ablation_agents}
\end{figure}

\end{document}